\documentclass[a4paper,UKenglish,cleveref, autoref, thm-restate,authorcolumns]{lipics-v2021}

\title{The Orbit Problem for Parametric Linear Dynamical Systems} %

\author{Christel Baier}{
Technische Universität Dresden, Germany	
}{christel.baier@tu-dresden.de%
}{https://orcid.org/0000-0002-5321-9343%
}{%
}

\author{Florian Funke}{
Technische Universität Dresden, Germany	
}{florian.funke@tu-dresden.de%
}{https://orcid.org/0000-0001-7301-1550%
}{%
}
\author{Simon Jantsch}{
Technische Universität Dresden, Germany	
}{simon.jantsch@tu-dresden.de%
}{https://orcid.org/0000-0003-1692-2408%
}{%
}

\author{Toghrul Karimov}{
Max Planck Institute for Software Systems, Saarland Informatics Campus, Germany
}{toghs@mpi-sws.org%
}{https://orcid.org/0000-0002-9405-2332%
}{%
}

\author{Engel Lefaucheux}{
Max Planck Institute for Software Systems, Saarland Informatics Campus, Germany
}{elefauch@mpi-sws.org%
}{https://orcid.org/0000-0003-0875-300X%
}{%
}

\author{Florian Luca}{
School of Mathematics, Wits University, Johannesburg, South Africa\and
Research Group in Algebraic Structures \& Applications, King Abdulaziz University,  Saudi Arabia\and
Max Planck Institute for Software Systems, Saarland Informatics Campus, Germany
}{Florian.Luca@wits.ac.za%
}{https://orcid.org/0000-0003-1321-4422%
}{%
}

\author{Jo\"el Ouaknine}%
{Max Planck Institute for Software Systems, Saarland Informatics Campus, Germany}%
{joel@mpi-sws.org}{https://orcid.org/0000-0003-0031-9356}{%
ERC grant AVS-ISS
(648701).
Also affiliated with Keble College, Oxford as \href{http://emmy.network/}{\texttt{emmy.network}} Fellow.}

\author{David Purser}{
Max Planck Institute for Software Systems, Saarland Informatics Campus, Germany	
}{dpurser@mpi-sws.org%
}{https://orcid.org/0000-0003-0394-1634
}{%
}

\author{Markus A.\,Whiteland}{
Max Planck Institute for Software Systems, Saarland Informatics Campus, Germany	
}{mawhit@mpi-sws.org%
}{https://orcid.org/0000-0002-6006-9902%
}{%
}

\author{James Worrell}{
Department of Computer Science, University of Oxford, UK
}{jbw@cs.ox.ac.uk%
}{https://orcid.org/0000-0001-8151-2443%
}{%
Supported by EPSRC Fellowship EP/N008197/1.
}

\authorrunning{C. Baier et al.} %

\Copyright{Christel Baier, Florian Funke, Simon Jantsch, Engel Lefaucheux, Florian Luca, Jo\"el Ouaknine, David Purser, Markus A.\,Whiteland and James Worrell} %

\ccsdesc[500]{Theory of computation~Logic and verification}

\keywords{Orbit problem, parametric, linear dynamical systems} %

\category{} %

\relatedversiondetails{Published Version}{https://doi.org/10.4230/LIPIcs.CONCUR.2021.28}

\funding{This work was funded by DFG grant 389792660 as part of TRR~248 -- CPEC	 (see \href{https://perspicuous-computing.science}{\tt perspicuous-computing.science}), the Cluster of Excellence EXC 2050/1 (CeTI, project ID 390696704, as part of Germany's Excellence Strategy), DFG-projects BA-1679/11-1 and BA-1679/12-1, and the Research Training Group QuantLA (GRK 1763).}

\nolinenumbers %

\hideLIPIcs  %

\EventEditors{Serge Haddad and Daniele Varacca}
\EventNoEds{2}
\EventLongTitle{32nd International Conference on Concurrency Theory (CONCUR 2021)}
\EventShortTitle{CONCUR 2021}
\EventAcronym{CONCUR}
\EventYear{2021}
\EventDate{August 23--27, 2021}
\EventLocation{Virtual Conference}
\EventLogo{}
\SeriesVolume{203}
\ArticleNo{28}

\usepackage{thmtools}
\usepackage[utf8]{inputenc}

\usepackage{graphicx}
\usepackage[full,bold]{complexity}

\usepackage{hyperref}

\newcommand{\naturals}{\mathbb{N}}

\newcommand{\reals}{\mathbb{R}}
\newcommand{\rationals}{\mathbb{Q}}
\newcommand{\algebraics}{\overline{\mathbb{Q}}}
\newcommand{\algebraicfunctions}{\mathbb{K}}

\renewcommand{\M}{\mathcal{M}}

\newcommand{\abr}[1]{\ensuremath{\langle #1 \rangle}}

\theoremstyle{remark}

\usepackage{etoolbox}
\AtBeginEnvironment{proof}{\setcounter{case}{0}}
\AtBeginEnvironment{claimproof}{\setcounter{case}{0}}

\usepackage{xcolor}

\def\({\left(}
\def\){\right)}
\def\<{\langle}
\def\>{\rangle}

\def\N{{\mathbb N}}
\def\Z{{\mathbb Z}}

\def\Q{{\mathbb Q}}
\def\R{{\mathbb R}}
\def\C{{\mathbb C}}

\newcommand{\ig}{\lambda}
\newcommand{\tg}{\gamma}
\newcommand{\eu}{\mathrm{e}}
\newcommand{\iu}{\mathrm{i}}

\newcommand{\source}{u}
\newcommand{\target}{v}
\newcommand{\param}{x}
\newcommand{\point}{s}
\newcommand{\exceptionalpoints}{\mathcal{E}}

\DeclareMathOperator{\rank}{rank}
\DeclareMathOperator{\ord}{ord}
\renewcommand{\Re}{\mathrm{Re}}
\renewcommand{\Im}{\mathrm{Im}}

\DeclareMathOperator\Log{Log}
\DeclareMathOperator{\lcm}{lcm}

\usepackage{todonotes}

\theoremstyle{definition}
\newtheorem{problem}[theorem]{Problem}

\begin{document}

\maketitle

\begin{abstract}
We study a parametric version of the Kannan-Lipton Orbit Problem for linear dynamical systems. We show decidability in the case of one parameter and Skolem-hardness with two or more parameters. 

More precisely, consider a $d$-dimensional square matrix $M$ whose entries are algebraic functions in one or more real variables. Given initial and target vectors $\source,\target\in \rationals^d$, the parametric point-to-point orbit problem asks whether there exist values of the parameters giving rise to a concrete matrix $N \in \reals{}^{d\times d}$, and a positive integer $n\in \naturals$, such that $N^n\source = \target$. 

We show decidability for the case in which $M$ depends only upon a single
parameter, and we exhibit a reduction from the well-known Skolem
Problem for linear recurrence sequences,
suggesting intractability in the case of two or more parameters.
\end{abstract}

\section{Introduction}

The \emph{Orbit Problem} for linear dynamical systems asks to decide, given a square matrix $M\in \rationals^{d\times d}$ and two vectors $\source,\target\in \rationals^d$, whether there exists a natural number $n$ such that $M^n\source=\target$. The problem was shown decidable (in polynomial time) by Kannan and Lipton~\cite{KannanL86} over ten years after Harrison first raised the question of decidability~\cite{Harrison69}.  The current paper is concerned with a generalisation of the Orbit Problem to \emph{parametric} linear dynamical systems.  In general, parametric models address a major drawback in quantitative verification, namely the unrealistic assumption that quantitative data in models are known \emph{a priori} and can be specified exactly.  In applications of linear dynamical systems to automated verification, parameters are used to model partially specified systems (e.g., a faulty component with an unknown failure rate, or when transition probabilities are only known up to some bounded precision) as well as to model the unknown environment of a system.  Interval Markov chains can also be considered as a type of parametric linear dynamical system.

\begin{problem}[Parametric Orbit Problem]\label{problem:main}  Given a
$(d\times d)$-matrix $M$, initial and target
vectors $\source,\target$, whose entries are
real algebraic functions in $\ell$ common real variables
$X = (\param_1,...,\param_\ell)$,  does there exist
$\point\in\mathbb{R}^\ell$, i.e., values of the parameters
giving rise to a concrete matrix, initial and target
$M(\point) \in \reals{}^{d\times d},\source(\point),\source(\point)\in \reals{}^{d}$, and a positive integer
$n\in \naturals$, such that $M(\point)^n \source(\point) = \target(\point)$?
\end{problem}
We prove two main results in this paper. In the case of a single parameter we show that the Parametric Orbit Problem is decidable. On the other hand, we show that the Parametric Orbit Problem is at least as hard as the
Skolem Problem---a well-known decision problem for linear recurrence
sequences, whose decidability has remained open for many decades. Our reduction establishes intractability in the case of two or more
parameters.

Thus our main decidability result is as follows:

\begin{restatable}{theorem}{restatethmmain}\label{thm:forplds}
  \cref{problem:main} is decidable when there is a single parameter (i.e., $\ell = 1$).
\end{restatable}

\Cref{thm:forplds} concerns a reachability problem in which the parameters are existentially quantified.  It would be straightforward to adapt our methods to allow additional constraints on the parameter, e.g., requiring that $s$ lie in a certain specified interval.  In terms of verification, a negative answer to an instance of the above reachability problem could be seen as establishing a form of robust safety, i.e., an `error state' is not reachable regardless of the value of the unknown parameter.

The proof of \Cref{thm:forplds} follows a case distinction based on properties of the eigenvectors of the matrix $M$ (whose entries are functions) and the shape of the Jordan normal form $J$ of $M$. Our theorem assumes the entries of the matrix, initial and target vectors are real algebraic functions---in particular encompassing polynomial and rational functions.  Note that even if we were to restrict the entries of $M$ to be polynomials in the parameters, we would still require (complex) algebraic functions in the Jordan normal form.  We assume a suitable effective representation of algebraic functions that supports evaluation at algebraic points, computing the range and zeros of the functions, arithmetic operations, and extracting roots of polynomials whose coefficients are algebraic functions.

The most challenging cases arise when $J$ is diagonal.  In this situation
we can reformulate the problem as follows: given algebraic functions
$\ig_i(\param), \tg_i(\param)$ for $1\leq i\leq t$, does there exist
$(n,\point)\in {\mathbb N}\times \mathbb{R}$ such that
\begin{equation}\label{eq:algfunctions}
	\ig_i^n(\point)=\tg_i(\point)\qquad {\text{\rm for~all}}\qquad i=1,\ldots,t?
\end{equation}

A further key distinction in analysing the problem in
\cref{eq:algfunctions} involves the rank of the
multiplicative group generated by the functions $\ig_1,\ldots,\ig_t$.
To handle the case that the group has rank at least two, a central
role is played by the results of Bombieri, Masser, and Zannier (see
\cite[Theorem 2]{BMZ1} and \cite{BMZ2}) concerning the intersection of
a curve in $\mathbb{C}^m$, with algebraic subgroups of
$(\mathbb{C}^*)^m$ of dimension at most $m-2$.  To apply these results
we view the problem in \cref{eq:algfunctions} geometrically
in terms of whether a curve
\[ C=
  \{(\lambda_1(s),\ldots,\lambda_t(s),\gamma_1(s),\ldots,\gamma_t(s))
  : s \in \mathbb{R} \} \subseteq \mathbb{C}^{2t}\] intersects the multiplicative group
\begin{equation*}
 G_n = \{(\alpha_1,\ldots,\alpha_t,\beta_1,\ldots,\beta_t) \in
  (\mathbb{C}^*)^{2t} \colon 
\alpha_1^n=\beta_1 \wedge \cdots \wedge \alpha_t^n=\beta_t \}
\end{equation*}
for some $n\in \mathbb{N}$.  The above-mentioned results of
Bombieri, Masser,  and Zannier can be used to derive an upper bound on
$n$ such that $C\cap G_n$ is non-empty under certain conditions on the
set of multiplicative relations holding among
$\lambda_1,\ldots,\lambda_t$ and $\gamma_1,\ldots,\gamma_t$.

We provide specialised arguments for a number of cases for which the
results of Bombieri, Masser, and Zannier cannot be applied.  In
particular, for the case that the multiplicative group generated by
the functions $\ig_1,\ldots,\ig_t$ has rank one, we provide in
Section~\ref{sec:rank1case} a direct elementary method to find
solutions of \cref{eq:algfunctions}.

Another main case in the proof is when matrix
$J$ has a Jordan block of size at least~2, i.e., it is not diagonal
(see Section~\ref{jordancell:dim2}).
The key
instrument here is the notion of the Weil height of an algebraic number
together with bounds that relate the height of a number to the height
of its image under an algebraic function.  Using these bounds we
obtain an upper bound on the $n\in\mathbb{N}$ such that the equation
$M(\point)^n \source(\point) = \target(\point)$ admits
a solution $s\in\mathbb{R}$.

\subsubsection*{Related work}
\label{sec:related}
Reachability problems in (unparametrized) linear dynamical systems
have a rich history. Answering a question by
Harrison~\cite{Harrison69}, Kannan and Lipton~\cite{KannanL86} showed
that the point-to-point reachability problem in linear dynamical
systems is decidable in {\sf PTIME}. They also noticed that the
problem becomes significantly harder if the target is a linear
subspace---a problem that still remains open, but has been solved for
low-dimensional instances \cite{ChonevOW13}. This was extended to
polytope targets in~\cite{ChonevOW15}, and later further generalized to
polytope initial sets in~\cite{AlmagorOW17}. Orbit problems have recently
been studied in the setting of rounding functions
\cite{Baieretal20}.
In our analysis we will make use of a version of the point-to-point reachability problem that allows matrix entries to be algebraic numbers.
In this case the eigenvalues are again algebraic, and decidability
follows by exactly the same argument as the rational case (although the algorithm is no longer in {\sf PTIME}), and is also a special case of the main result of \cite{cai2000complexity}.

If the parametric matrix $M$ is the transition matrix of a parametric
Markov chain (pMC) \cite{JonssonL1991, GivanLD2000, KozineU2002}, then
our approach combines \emph{parameter synthesis} with the
\emph{distribution transformer semantics}. Parameter synthesis on pMCs
asks whether some (or every) parameter setting results in a Markov
chain satisfying a given specification, expressed, e.g., in PCTL
\cite{JungesAHJK19}. An important problem in this direction is to find
parameter settings with prescribed properties
\cite{LanotteMT2007,CeskaDKP14,CubuktepeJJKT18}, which has also been
studied in the context of model repair
\cite{BartocciGKRS11,PathakAJTK15}.  While all previous references use
the standard path-based semantics of Markov chains, the distribution
transformer semantics \cite{KwonA04,KorthikantiVAK10,ChadhaKVAK11}
studies the transition behaviour on probability distributions. It has, to the best of our knowledge,
never been considered for
parametric Markov chains. Our approach implicitly does this in that it
performs parameter synthesis for a reachability property in the
distribution transformer semantics.

The Skolem Problem asks whether a linear recurrence sequence $(u_n)_n$ has a zero term ($n$ such that $u_n = 0$). Phrased in terms of linear dynamical systems, the Skolem Problem asks whether a $d$-dimensional linear dynamical system hits a $(d-1)$-dimensional hyperplane, and decidability in this setting is known for matrices of dimension
at most four \cite{MignotteST84,Vereshchagin85}. A continuous version
of the Skolem Problem was examined in \cite{ChonevOW16b}. With the
longstanding intractability of the Skolem Problem in general, it has
recently been used as a reference point for other decision problems
\cite{AkshayATOW15,MajumdarMS20,PiribauerB20}.

Ostafe and Shparlinski~\cite{OS20} consider the Skolem
Problem for parametric families of simple linear recurrences.
More precisely, they consider linear recurrences of the form
$u_n = a_1(\param)\lambda_1(\param)^n + \cdots + a_k(\param)\lambda_k^n(\param)$ for
rational functions $a_1,\ldots,a_k,\lambda_1,\ldots,\lambda_k$
with coefficients in a number field.  They show that the existence of
a zero of the sequence $(u_n)$ can be decided for all values of the
parameter outside an exceptional set of numbers of bounded height
(note that any value of the parameter such that the sequence $u_n$ has a
zero is necessarily algebraic).

\section{Preliminaries}
\label{sec:prelim}

We denote by $\mathbb{R},\C,\Q,\algebraics$ the real, complex, rational, and algebraic numbers respectively.
For a field $K$ and a finite set $X$ of variables, $K[X]$ and $K(X)$ respectively denote the ring of polynomials and field of rational functions with coefficients in $K$. 
A meromorphic function%
\footnote{A ratio of two holomorphic functions, which are complex-valued functions complex differentiable in some neighbourhood of every point of the domain.} $f\colon U \to \C$ where $U$ is some open subset $U \subseteq \C^{\ell}$ is called algebraic, if
$P(x_1,\ldots,x_{\ell},f(x_1,\ldots,x_{\ell})) = 0$ for some
$P \in \Q[x_1,\ldots,x_{\ell},y]$. We say that
$f$ is \emph{real algebraic} if it is real-valued on real inputs.

\begin{definition}\label{def:PMC}
	A \emph{parametric Linear Dynamical System} (pLDS) of \emph{dimension} $d\in\N$ 
	is a tuple $\M = (X, M, \source)$, where 
		 $X$ is a finite set of \emph{parameters},
		 $M$ is the \emph{parametrized matrix} whose entries are real algebraic functions in parameters $X$ and
		 $\source$ is the parametric initial distribution whose entries are also real algebraic functions in parameters $X$.
\end{definition}
Given $\point\in \R^{|X|}$, we denote by $M(\point)$ the matrix $\R{}^{d\times d}$ obtained 
from $M$ by evaluating each function in $M$ at $\point$, provided that 
this value is well-defined. Likewise we obtain $\source(\point)$. 
We call $(M(\point),\source(\point))$ the induced linear dynamical system (LDS). The \emph{orbit} of the LDS
$(M(\point),\source(\point))$ is the
set of vectors obtained by repeatedly applying the matrix $M(\point)$ to $\source(\point)$: $\{\source(\point),M(\point)\source(\point),M(\point)^2\source(\point),\dots\}$. 
The LDS $(M(\point),\source(\point))$ \emph{reaches} a target $\target(\point)$ if $\target(\point)$ is in the orbit, \emph{i.e.} there exists $n\in\N$ such that $M(\point)^n\source(\point) = \target(\point)$.

We remark that $M(\point)$ is undefined whenever any of the entries of $M$ is undefined. 
For any fixed $n$, the elements of $M^n$ are polynomials in the entries of
$M$, and consequently,
$M^n$ is defined on the same domain as $M$.

Unless we state that $M$ is a constant function, all matrices should be seen as functions, with parameters $\param_1,\dots,\param_{|X|}$, or simply $\param$ if there is a single parameter. The notation $\point$ is used for a specific instantiation of $\param$. We often omit $\param$ when referring to a function, either   the function is declared constant or when we do not need to make reference to its parameters.

\subsection{Computation with algebraic numbers}

Throughout this note we employ notions from (computational) algebraic geometry and algebraic number theory.
Our approach relies on transforming the matrices we consider in Jordan normal form. Doing so, the coefficients of the computed matrix are not rational anymore but algebraic. Next we recall the necessary basics and refer to \cite{CohenBook,Waldschmidt2000Diophantine} for more background
on notions utilised throughout the text.

The algebraic numbers $\algebraics{}$ are the complex numbers which can be
defined as some root of a univariate polynomial in $\Q[\param]$. In
particular, the rational numbers are algebraic numbers. For every $\alpha \in \algebraics$ there exists a unique monic univariate polynomial $P_{\alpha} \in \Q[\param]$
of minimum degree for which $P_{\alpha}(\alpha) = 0$. We call $P_{\alpha}$
the \emph{minimal polynomial} of $\alpha$. An algebraic number $\alpha$ is represented
as a tuple $(P_{\alpha},\alpha^*,\varepsilon)$, where $\alpha^*=a_1+a_2i$, $a_1,a_2\in\mathbb{Q}$, is an
approximation of $\alpha$, and $\varepsilon \in \mathbb{Q}$ is
sufficiently small such that $\alpha$ is the unique root of $P_{\alpha}$ within
distance $\varepsilon$ of $\alpha^*$ (such $\varepsilon$ can be computed by the root-separation bound, due to Mignotte \cite{Mig82}).  This is referred to as the \emph{standard}
or \emph{canonical representation} of an algebraic number.
Given canonical representations of two algebraic numbers $\alpha$ and
$\beta$, one can compute canonical representations of $\alpha+\beta$,
$\alpha\beta$, and $\alpha/\beta$, all in polynomial time.

\begin{definition}[Weil's absolute logarithmic height]
Given an algebraic number $\alpha$ with minimal polynomial $p_{\alpha}$ of degree $d$, consider the polynomial $a_d p_{\alpha}$ with $a_d \in \N$ minimal such that for $a_d p_\alpha = a_d \param^d+\cdots+a_1 \param + a_0$ we have $a_i \in \Z$ and $\gcd(a_1,\ldots,a_d) = 1$. Write $a_d p_{\alpha} = a_d(\param-\alpha^{(1)})\cdots (\param-\alpha^{(d)})$,
 where $\alpha^{(1)}=\alpha$. Define the \emph{(Weil) height} $h(\alpha)$ of $\alpha \neq 0$ by $
 h(\alpha)=\frac{1}{d}\Big(\log a_d + \sum_{i=1}^d \log(\max\{|\alpha^{(i)}|,1\})\Big)$. By convention $h(0) = 0$.
\end{definition}
For all $\alpha,\beta \in \overline{\Q}$ and
$n \in \Z$ we have from \cite[Chapt.~3]{Waldschmidt2000Diophantine}:
 \begin{enumerate}
 \item  $h(\alpha + \beta)\le h(\alpha)+h(\beta)+\log 2;$
\item  $h(\alpha \beta)\le h(\alpha)+h(\beta)$;
\item  $h(\alpha^n)=|n|\cdot h(\alpha)$.
\end{enumerate}
In addition, for $\alpha \neq 0$ we have
$h(\alpha)=0$ if and only if $\alpha$ is a root of unity ($\alpha$ is a root of unity if there exists $k\in\naturals{}$, $k\geq 1$, such that $\alpha^k = 1$). Notice that the set of algebraic numbers with both height and degree bounded is always finite.

\subsection{Univariate algebraic functions}
\label{sec:AF}

Let $K$ be an algebraic extension of a field $L$ such that the characteristic polynomial of $M \in L^{d\times d}$ splits into linear factors over $K$. It is well-known that we can factor $M$ over
$K$ as
$M=C^{-1}JC$ for some
invertible matrix $C \in K^{d\times d}$ and block diagonal Jordan matrix
$J=\langle{J_{1},\ldots,J_{N}}\rangle \in K^{d\times d}$. Each block $J_{i}$ associated with some
eigenvalue $\lambda_i$, and $J_{i}^{n}$, have the following \emph{Jordan block} form for some $k\geq 1$:
\begin{equation*}
J_{i} = \left(\begin{smallmatrix}
\lambda	&&	1		&&	0		&&	\cdots	&&	0		\\
0		&&	\lambda	&&	1		&&	\cdots	&&	0		\\
\vdots	&&	\vdots	&&	\vdots	&&	\ddots	&&	\vdots	\\
0		&&	0		&&	0		&&	\cdots	&&	1		\\
0		&&	0		&&	0		&&	\cdots	&&	\lambda	\\
\end{smallmatrix}\right) \qquad \text{ and } \qquad
J_{i}^{n}=\left(\begin{smallmatrix}
\lambda^{n}	&&	n\lambda^{n-1}	&& \binom{n}{2}\lambda^{n-2}	&&
\cdots		&&	\binom{n}{k-1}\lambda^{n-k+1}				\\
0			&&	\lambda^{n}		&&	n\lambda^{n-1}				&&
\cdots		&&	\binom{n}{k-2}\lambda^{n-k+2}				\\
\vdots	&&	\vdots	&&	\vdots	&&	\ddots	&&	\vdots			\\
0		&&	0		&&	0		&&	\cdots	&&	n\lambda^{n-1}	\\
0		&&	0		&&	0		&&	\cdots	&&	\lambda^{n}		\\
\end{smallmatrix}\right).
\end{equation*}
Furthermore, each eigenvalue $\ig$ of $M$
appears in at least one of the Jordan blocks.

In case $L = \Q$, we may take $K$ to be an algebraic number field. In particular, the eigenvalues of a rational matrix are algebraic. %
However, in this paper, the entries of our matrix are \emph{algebraic functions}, and so too are the entries in Jordan normal form.
We recall some basics of algebraic geometry and univariate algebraic functions required for the analysis in the single-parameter setting, and refer the reader to \cite{BasuBook,CoxLOS1991ideals} for further information.

Let $U \subseteq \mathbb{C}$ be a connected open set and
$f:U\rightarrow \mathbb{C}$ a meromorphic function. %
We say that $f$
is \emph{algebraic over $\mathbb{Q}(x)$} if there is a polynomial
$P(x,y) \in \mathbb{Q}[x,y]$ such that $P(x,f(x))=0$ for all $x \in U$
where $f$ is defined.  Notice that a univariate algebraic function has finitely
many zeros and poles, and furthermore, these zeros and poles (or zeros at $\infty$) are
algebraic. Indeed, let $P(x,y) = a_d(x)y^d + \dots + a_1(x)y+ a_0(x)$,
with $a_i \in \Q[x]$, be irreducible.  Assuming that $f$ vanishes at $s$, we have that
$a_0(s) = 0$. There are only finitely many $s$ for which this can
occur. Furthermore, the function $1/f$ is meromorphic (on a possibly different domain $U$) and satisfies
$y^{d}P(x,1/y) = a_d(x) + \ldots + a_1(x)y^{d-1} + a_0(x) y^{d}$. We conclude that a pole of $f$ (a zero of $1/f$) is a zero of $a_d(x)$.

Let $P(x,y) = \sum_{i=0}^d a_i(x)y^i \in \mathbb{Q}(x)[y]$. 
We say that $c \in
\mathbb{C}$ is a \emph{critical point} of $P$ if either $a_d(c)=0$ or 
the resultant $\mathrm{Res}_y(P,\frac{\partial P}{\partial y})$ 
vanishes at $c$.  If $P$ is irreducible, then it has only finitely many
critical points since the resultant is a univariate non-zero polynomial.

Let $M$ be a $(d\times d)$-matrix with univariate real algebraic functions as entries. Let its characteristic polynomial be $P(x,y):= \det(Iy-M)$ and write $c_1,\ldots,c_m \in \mathbb{C}$ for the
critical points of the irreducible factors of $P$. Then
there exist a connected open subset $U\subseteq \mathbb{C}$ such that
$\mathbb{R}\setminus\{c_1,\ldots,c_m\} \subseteq U$, and $d$
holomorphic functions
$\lambda_1,\ldots,\lambda_d : U \rightarrow \mathbb{C}$ (not
necessarily distinct) such that the characteristic polynomial
$P$ of $M$ factors as
\[ P(x,y) = (y-\lambda_1(x))(y-\lambda_2(x)) \cdots
  (y-\lambda_d(x)) \] for all points $x \in U$
  (see, e.g., \cite[Chapt.~1, Thm.~8.9]{Foster81compact}).

Let us fix a $(d\times d)$-matrix $M$ and vectors $u$, $v$ with univariate real algebraic entries. We thus have $M \in L^{d\times d}$,
$u,v \in L^d$, for some finite field extension $L$ of $\Q(x)$. Let $\mathbb{K}$ be fixed to an algebraic extension of $L$ such that the characteristic polynomial of $M$ splits into linear factors over the field
$\mathbb{K}$. 
Then, over the
field $\mathbb{K}$ we have the factorisation $M=C^{-1}JC$ with $J$
in Jordan form.  The eigenvalues of $M$, denoted $\ig_1, \ldots, \ig_k$, appear in the diagonal of $J$. Let the set of exceptional points, denoted $\exceptionalpoints$, consist of the finite set $\{c_1,\dots,c_m\}$, the poles of the entries of $M,C,C^{-1},J,u$ and $v$, and points where $\det C(s) = 0$ (i.e., $C(s)$ is singular).

Consider now a non-constant univariate algebraic function $\ig$ not necessarily real.
In our analysis, we shall need to bound the height $h(\ig(s))$ in terms of $h(s)$,
as long as $s$ is not a zero or a pole of $\ig$. The following lemma shows $h(\ig(s)) = \Theta(h(s))$:
\begin{restatable}{lemma}{boundhlsbyhs}\label{lem:1}
Let $\ig$ be a non constant algebraic function in $\mathbb{K}$. Then there exist effective constants $c_1,c_2,c_3,c_4 > 0$ such that for algebraic $s$ not a zero or pole of $\ig$ we have\\
$c_1h(s) - c_2 \le  h(\ig(s)) \le  c_3h(s) + c_4.$
\end{restatable}

\subsubsection{Multiplicative relations}
Let $Y = \{\ig_1,\ldots,\ig_t\} \subset \mathbb{K}$ be a set of univariate algebraic functions.
\begin{definition}
 A tuple $(a_1,\ldots,a_t)\in \Z^t$ for which $\ig_1^{a_1} \cdots \ig_t^{a_t} = 1$ identically, is called
  a \emph{multiplicative relation}. A set of multiplicative relations is called \emph{independent} if it is $\Z$-linearly independent as a subset of $\Z^t$. The set $Y$ is said to be \emph{multiplicatively dependent}
  if it satisfies a non-zero multiplicative relation. Otherwise $Y$
  is \emph{multiplicatively independent}. 
\label{def:rank}The \emph{rank} of $Y$, denoted $\rank Y$, is the size of the largest
multiplicatively independent subset of $Y$.

A tuple $(a_1,\ldots,a_t) \in \Z^t$, for which there exists $c\in \overline{\Q}$ such that $\ig_1^{a_1}\cdots \ig_t^{a_t} = c$ identically, is called a \emph{multiplicative relation modulo constants}. We say that $Y$ is \emph{multiplicatively dependent modulo constants}
if it satisfies a non-zero multiplicative relation modulo constants. Otherwise $Y$ is \emph{multiplicatively independent
modulo constants}.
\end{definition}

In particular, if $\rank\langle\ig_1,\ldots,\ig_t\rangle = 1$, then
for each pair $\ig_i$, $\ig_j$, we have $\ig_i^{b} = \ig_j^{a}$ for
some integers $a$, $b$ not both zero. In the analysis that follows, we
only need to distinguish between this case and
$\rank\langle\ig_1,\ldots,\ig_t\rangle \ge 2$. We will also need to
find multiplicative relations modulo constants between algebraic
functions.
These can be algorithmically determined and constructed as a consequence of the following proposition.
To this end, let $L$ and $L' \subseteq \Z^t$ be the set of multiplicative relations and multiplicative relations modulo constants on $Y$, respectively.
Both $L$ and $L'$ are finitely generated as
subgroups of $\Z^t$ under vector addition.

\begin{restatable}{proposition}{multiplicative}
Given a set $Y = \{\ig_1,\ldots,\ig_t\}$ of univariate algebraic functions, one can compute
a generating set for both $L$ and $L'$.%
\label{prop:multiplicative}
\end{restatable}
\begin{proof}%
This is essentially a special case of a result from \cite{derksen2005quantum}. Indeed, in Sect.~3.2,
they show how to find the generators of the group $L$ in case the $\ig_i$ are elements of a finitely generated field over $\Q$. We apply the result to the field
$\Q(x,\ig_1,\ldots,\ig_t)$ to obtain the claim
for the set $L$. For $L'$, Case 3 of
\cite[Sect.~3.2]{derksen2005quantum} computes a generating
set as an intermediate step in the computation of a basis of
$L$.
Specifically, $L$ and $L'$ are the respective kernels of
  the maps $\varphi$ and
  $\tilde{\varphi}$ in~\cite[Sect.~3.2]{derksen2005quantum}. We give an alternative proof sketch specialised to univariate functions in~\cref{app:prelim}.
\end{proof}

\section{The Multi-Parameter Orbit Problem is Skolem-hard }\label{sec:hardness}

The \emph{Skolem Problem} asks, given a order-$k$ linear recurrence sequence $(u_n)_n$, uniquely defined by a recurrence relation $u_{n} = a_1 u_{n-1} + \dots + a_k u_{n-k}$ for fixed $a_1, \ldots, a_k$ and initial points $u_1,\dots,u_k$, whether there exists an $n$ such that $u_n = 0$.  The problem is famously not known to be decidable for orders at least 5, and problems which the Skolem problem reduce to are said to be \emph{Skolem-hard}. We will now reduce the Skolem at order 5 to the two-parameter parametric orbit problem. 

It suffices to only consider the instances of Skolem Problem at order 5 of  the form $u_n = a\lambda_1^n + \overline{a \lambda_1^n} + b \lambda_2^n + \overline{b \lambda_2^n} + c \rho^n = 0$ with $|\lambda_1| = |\lambda_2| \geq |\rho|$ and $a, b, \lambda_1, \lambda_2 \in \algebraics$, $c, \rho \in \algebraics \cap \R$, as the instances of the Skolem Problem at order 5 that are not of this form are known to be decidable \cite{OuaknineW2012decision}. We may assume that $c = \rho = 1$ by considering the sequence $(u_n/c\rho^n)$ if necessary. We can also rewrite $u_n = A \Re{\lambda_1^n} + B \Im{\lambda_1^n} + C\Re{\lambda_2^n} + D \Im{\lambda_2^n} + 1$ for $A, B, C, D \in \algebraics \cap \R$. %

Let $u_n = a\lambda_1^n + \overline{a \lambda_1^n} + b \lambda_2^n + \overline{b \lambda_2^n} + 1 = A \Re{\lambda_1^n} + B \Im{\lambda_1^n} + C\Re{\lambda_2^n} + D \Im{\lambda_2^n} + 1$ be a hard instance of the Skolem Problem. Let $M = \operatorname{diag}\left(
\begin{bmatrix} 
\Re{\lambda_1} & -\Im{\lambda_1}\\
\Im{\lambda_1} & \Re{\lambda_1}
\end{bmatrix},
\begin{bmatrix} 
\Re{\lambda_2} & -\Im{\lambda_2}\\
\Im{\lambda_2} & \Re{\lambda_2}
\end{bmatrix} 
\right)$, that is, the Real Jordan Normal Form of $\operatorname{diag}(\lambda_1, \overline{\lambda_1}, \lambda_2, \overline{\lambda_2})$. We set the starting point to be $u = [1 \text{ } 1 \text{ } 1 \text{ } 1]^\top$ and show how to define parametrized target vectors $v_1(s,t), \ldots, v_k(s,t)$ such that for all $n$, $u_n = 0$ if and only if there exist $s,t \in \R$ such that $M^n u = v_i(s,t)$ for some $i$. The Skolem Problem at order 5 then reduces to $k$ instances of the two-parameter orbit problem. 

The idea of our reduction is to first construct a semiagebraic set $Z \subseteq \R^4$, $Z = \bigcup_{i=1}^{k} Z_i$ such that $u_n = 0$ if and only if $(\Re{\lambda_1^n}, \Im{\lambda_1^n}, \Re{\lambda_2^n},\Im{\lambda_2^n}) \in Z$, and each $Z_i$ is a semialgebraic subset of $\R^4$ that can be described using two parameters and algebraic functions in two variables. Observing that $M^ns = (Re{\lambda_1^n} -\Im{\lambda_1^n}, \Im{\lambda_1^n} + \Re{\lambda_1^n}, Re{\lambda_2^n} -\Im{\lambda_2^n}, \Im{\lambda_2^n} + \Re{\lambda_2^n})$,
we then compute $v_i(s,t)$ from $Z_i$ as follows. Suppose $Z_i = \{(x(s,t), y(s,t), z(s,t), u(s,t) : s, t. \in \R\}$. Then $v_i(s,t) = (x(s,t) - y(s,t), y(s,t) + x(s,t), u(s,t)-v(s,t), v(s,t) + u(s,t))$.

To compute $Z$, first observe that $\Im{\lambda_2^n} = \pm\sqrt{(\Re{\lambda_1^n})^2 + (\Im{\lambda_1^n})^2 - (\Re{\lambda_2^n})^2}$ for all $n$ as $|\lambda_1| = |\lambda_2|$. Motivated by this observation, let $S_+, S_- \subseteq \R^3$, $S_+= \{(x, y, z) : Ax + By +Cz + D\sqrt{x^2 + y^2 - z^2} + 1 = 0\}$ and $S_-= \{(x, y, z) : Ax + By +Cz - D\sqrt{x^2 + y^2 - z^2} + 1 = 0\}$. We will choose $Z = \{ (x,y,z,\sqrt{x^2 + y^2 - z^2}) : (x,y,z) \in S_+ \} \cup \{ (x,y,z,-\sqrt{x^2 + y^2 - z^2}) : (x,y,z) \in S_- \}$. It is easy to check that the above definition of $Z$ satisfies the requirement that $u_n = 0$ if and only if $(\Re{\lambda_1^n}, \Im{\lambda_1^n}, \Re{\lambda_2^n},\Im{\lambda_2^n}) \in Z$, and it remains to show that both $S_+$ and $S_-$ can be parametrized using algebraic functions in two variables and two parameters. To this end, observe that $S_+$ and $S_-$ are both semialgebraic subsets of $\R^3$, but are also contained in the algebraic set $S = \{(x,y,z) : (Ax + By +Cz + 1)^2 = D^2(x^2 + y^2 - z^2)\} \subseteq \R^3$. Since $S \neq \R^3$ (for example, $(0,0,0) \notin S$), and it is algebraic, $S$ can have \emph{dimension} (see \cite{CoxLOS1991ideals} for a definition) at most $2$. Hence $S_+, S_-$ also have semialgebraic dimension
at most $2$. In~\cref{appen:hardness}, we show that a semialgebraic subsets of $\R^3$ of dimension at most two can be written as a finite union of sets of the form $\{v(s,t) : s, t \in \R\}$, where $v$ is an algebraic function. This completes the construction of $Z$ and the description of the reduction.

\section{Single Parameter Reachability: Overview of proof}

In this section we show how to prove \autoref{thm:forplds}, that is, it is decidable, given a
$(d\times d)$-matrix $M$, initial and target
vectors $\source,\target$, whose entries are
real algebraic functions all depending on a single
  parameter,  whether there exist
$\point\in\mathbb{R}$
giving rise to a concrete matrix, initial and target
$M(\point) \in \reals{}^{d\times d},\source(\point),\target(\point)\in \reals{}^{d}$, and a positive integer
$n\in \naturals$, such that $M(\point)^n \source(\point) = \target(\point)$.

In our case analysis, we often  show that either there is a finite set of parameter values for which the constraints could hold, or place an upper bound on the $n$ for which the constraints hold. The following proposition shows that the decidability of the problem in these cases is apparent:
\begin{proposition}\label{lemma:fixedn}\label{sec:fixeds}\hfill
\begin{itemize}
\item Given a finite set $S\subset\mathbb{R}$ it is decidable if there exists $(n,s) \in \mathbb{N}\times S$ s.t. $M(\point)^n \source(\point) =\target(\point)$.
\item Given $B\in\mathbb{N}$ it is decidable if there exists $n\le B$ and $s \in \mathbb{R}$ s.t. $M(\point)^n \source(\point) =\target(\point) $.
\end{itemize}
\end{proposition}
\begin{proof}
 The decidability of the first case is a consequence of the fact that a choice of parameter leads to a concrete matrix, thus giving an instance of the non-parametric Orbit Problem.

In the second case, for fixed $n$, one can observe that the matrix $M^n$ is itself a matrix of real algebraic functions. Hence the equation $M^n\source =\target$ can be rewritten as equations $P_i(\param) =0$ for real algebraic $P_i$ for $i = 1,\dots,d$. For each equation the function is either identically zero, or vanishes at only finitely many $\point$ which can be determined, and one can check if there is an $\point$ in the intersection of the zero sets as $i$ varies. Repeat for each $n \le B$.
\end{proof}

As a consequence, for each $n$ either $M^n \source =\target$ holds identically (for every $s$), or there are at most finitely many $s$ such that $M(\point)^n \source(\point)= \target(\point)$, and all such points are algebraic, as they must be the roots of the algebraic functions $P_i$.

Our approach will be to place the problem into Jordan normal form (\cref{sec:parametricjordan}), where we will observe that the problem can be handled  if the resulting form is not diagonal (\cref{jordancell:dim2}). Here the relation between the Weil height of an algebraic number and its image under an algebraic function are exploited to bound $n$ (reducing to the second case of the proceeding proposition). 

In the diagonal case the problem can be reformulated for algebraic functions $\ig_i, \tg_i$ for $i = 1\dots,t$, whether there exist $(n,\point)\in {\mathbb N}\times\mathbb{R}\setminus \exceptionalpoints$ such that $ \ig_i^n(\point)=\tg_i(\point)$ for~all $ i=1,\ldots,t$,
where $\exceptionalpoints{}$ is a finite set of exceptional points.  These exceptional points can be handled separately using the first case of the proceeding proposition.

To show decidability we will distinguish between the case where $\rank\langle \ig_1,\ldots,\ig_t\rangle$ is 1 and when it is greater than 2 (recall \cref{def:rank}). As discussed in the introduction, the most intriguing part of our development will be in the case of $\rank\langle \ig_1,\ldots,\ig_t\rangle \ge 2$, captured in the following lemma:

\begin{restatable}{lemma}{mainlemma}\label{lemma:main}
Let $\ig_1,\ldots,\ig_t$ be algebraic functions in $\mathbb{K}$ and $\rank\langle \ig_1,\ldots,\ig_t\rangle\ge 2$. Given algebraic functions $\tg_1,\ldots,\tg_t$ in $\mathbb{K}$, then it is decidable
whether there exist $(n,\point)\in {\mathbb N}\times \R\setminus \exceptionalpoints$ such that
\begin{equation}\label{eq:mainlemmaequations}
 \ig_i(\point)^n=\tg_i(\point)\qquad {\text{\rm for~all}}\qquad i=1,\ldots,t.
 \end{equation}
\end{restatable}
The proof of this lemma is shown in \cref{sec:rank2}. Here we apply two specialised arguments,  in the case of non-constant $\lambda$'s we exploit the results of Bombieri, Masser, and Zannier~\cite{BMZ1,BMZ2} to show there is a finite effective set of parameter values. In the case of constant $\lambda$'s we reduce to an instance of Skolem's problem that we show is decidable, effectively bounding $n$.

It will then remain to prove a similar lemma for the case where the rank is 1. Here we will exploit the initial use of real algebraic functions, to ensure the presence of complex conjugates.

\begin{restatable}{lemma}{rankonelemma}\label{lemma:rank1case}
Let $\ig_1,\ldots,\ig_t$ be algebraic functions in $\mathbb{K}$ and $\rank\langle \ig_1,\ldots,\ig_t\rangle = 1$. We assume that, if $\ig_i$ is complex then $\overline{\ig_i}$ (the complex conjugate) also appears. Given algebraic functions $\tg_1,\ldots,\tg_t$ in $\mathbb{K}$, then it is decidable
whether there exist $(n,\point)\in {\mathbb N}\times\mathbb{R}\setminus \exceptionalpoints$ such that $\ig_i^n(\point)=\tg_i(\point)$ for all $ i=1,\ldots,t$.
\end{restatable}

The proof of this lemma (in \cref{sec:rank1case}), reduces the problem to a single equation $(t=1)$, for which we provide a specialised analysis on the behaviour of such functions that enable us to decide the existence of a solution.

In the remainder of this section we will show how to place the problem in the form of these two lemmas: first placing the  matrix into Jordan normal form, eliminating the cases where the Jordan form is not diagonal and provide some simplifying assumptions for the proofs of \cref{lemma:main,lemma:rank1case}.

\subsection{The parametric Jordan normal form}\label{sec:parametricjordan}

 For every $\point\in \mathbb{R}\setminus \exceptionalpoints$
we have
$M(\point)=C^{-1}(\point)J(\point)C(\point)$ and hence, for every $n\in\mathbb{N}$,
$M^n(\point) \source(\point)=\target(\point)$ if and only if
$J^n(\point)C(\point)\source(\point) = C(\point)\target(\point)$.
On the other hand, deciding whether there exists $\point\in \exceptionalpoints$ with
$M^n(\point) \source(\point)=\target(\point)$ reduces to finitely many instances of the
Kannan-Lipton Orbit Problem, which can be decided separately.
We have thus reduced the parametric point-to-point reachability problem
to the following one in case of a single parameter:

\begin{problem}\label{problem:afterjordan}
Given a matrix $J \in \mathbb{K}^{d\times d}$ in Jordan normal form, and
vectors $\tilde{\source}$, $\tilde{\target} \in \mathbb{K}^d$, decide whether there
exists $(n,\point) \in \N \times \mathbb{R}\setminus \exceptionalpoints$ such that $J^n(\point) \tilde{\source}(\point) = \tilde{\target}(\point)$.
\end{problem}

\begin{example}\label{ex:reachabilityJNF}
\label{ex:JNF}
Define
$M =
\left( \begin{smallmatrix}
 \param+\frac{1}{2} & 0 & 0 \\
 \frac{1}{2}-\param & 1-\param & 0 \\
 0 & \param & 1
\end{smallmatrix}
\right) \in \mathbb{Q}(x)^{3\times 3}$. Then the characteristic polynomial of $M$ is
$\det(yI -M) =   (y -1/2 - x) ( y - 1) (y + x -1)$. The irreducible factors have no critical points.
Now over $\mathbb{K}$ we may write $M = C^{-1} J C$, where
$J = \left(
\begin{smallmatrix}
 1 & 0 & 0 \\
 0 & 1-\param & 0 \\
 0 & 0 & \param+\frac{1}{2}
\end{smallmatrix}
\right)$,
$C = \left(
\begin{smallmatrix}
 1 & 1 & 1 \\
 \frac{1-2 \param}{4 \param-1} & -1 & 0\\
 \frac{2 \param}{1-4 \param} & 0  & 0 
\end{smallmatrix}
\right)$,
and
$C^{-1} = 
\left(
\begin{smallmatrix}
 0 & 0  & \frac{1}{2\param}-2 \\
 0 & -1 & 1 - \frac{1}{2\param} \\
 1 &  1 & 1 
\end{smallmatrix}
\right)$. Notice that $J$ is defined for all $\param$, while $C$ is not defined at $1/4$, and $C^{-1}$ is not defined at $0$ (notice also that $C(0)$ is not invertible). Therefore $\mathcal{E} = \{0,1/4\}$.
For $\point \in \R \setminus\cal E$, all three are defined and
we have $M(\point) = C^{-1}(\point) J(\point) C(\point)$, with $J(\point)$ in Jordan normal form
and $C(\point)$ invertible.

Notice, for $1/4 \in \cal E$, we have $M(1/4) = R^{-1} K R$, where
$K = \left(
\begin{smallmatrix}
 1 & 0 & 0 \\
 0 & \frac{3}{4} & 1 \\
 0 & 0 & \frac{3}{4}
\end{smallmatrix}
\right)$
and $R = \left(
\begin{smallmatrix}
 1 & 1 & 1 \\
 -1 & -1 & 0 \\
 -\frac{1}{4} & 0 & 0 \\
\end{smallmatrix}
\right)$. Notice here that $M(1/4)$ is non-diagonalisable (over $\algebraics{}$),
though $M$ is (over $\mathbb{K}$).

Let $\source = (\source_1,\source_2,\source_3) \in \mathbb{Q}(x)^3$ and $\target= (\target_1,\target_2,\target_3)\in \mathbb{Q}(x)^3$.
The problem of whether there exists $(n,\point)\in \N \times \R$ for which
$M(\point)^n \source(s) = \target(s)$ is reduced to checking the problem at $s \in \mathcal{E}$,
and to the associated problem $J^n(s) \tilde{\source}(s) = \tilde{\target}(s)$,
where $\tilde{\source} = \left(\begin{smallmatrix}
\source_1 + \source_2 + \source_3\\
\frac{ 1-2\param}{4 \param-1}\source_1 - \source_2\\
\frac{2\param}{1-4\param}\source_1
\end{smallmatrix}\right)$,
$\tilde{\target} = 
\left(\begin{smallmatrix}
\target_1 + \target_2 + \target_3\\
\frac{ 1-2\param}{4 \param-1}\target_1 - \target_2\\
\frac{2\param}{1-4\param}\target_1
 \end{smallmatrix}\right)
$,
and
$J^n = \left(\begin{smallmatrix}
 1 & 0 & 0 \\
 0 & (1-\param)^n & 0 \\
 0 & 0 & (\param+\frac{1}{2})^n
\end{smallmatrix}
\right)$.
\end{example}

Let us establish some notation: assume $J= \abr{J_1,\dots,J_N}$, corresponding to eigenvalues $\ig_1,\dots,\ig_N$. Assume the dimension of Jordan block $J_i$ is $d_i$, and let $\tilde{u}_{i,1},\dots,\tilde{u}_{i,d_i}$ be the coordinates of $\tilde{u}$ associated with the Jordan block $J_i$, where index $1$ corresponds to the bottom of the block. Similarly, let $\tilde{v}_{i,1},\dots,\tilde{v}_{i,d_i}$ be the corresponding entries of the target.

Let us define the functions $\gamma_1,\dots, \gamma_N$ used in our reduction to \autoref{lemma:main} and \autoref{lemma:rank1case}. We let $\gamma_i(s) = \tilde{v}_{i,1}(s)/\tilde{u}_{i,1}(s)$, for $\tilde{u}_{i,1}(s) \ne 0$. If $\tilde{\source}_{i,1}$ is not constant zero, then there are finitely many $s$ where $\tilde{\source}_{i,1}(\point) =0 $, each of which can be handled explicitly. If some $\tilde{\source}_{i,1}$ is the constant zero function, then there are two cases. Firstly, if $\tilde{\target}_{i,1}$ is also the constant zero then we are in the degenerate case $\lambda_i^n\cdot 0 =  0$, and the row can be ignored. Secondly if $\tilde{\target}_{i,1}$ is not constant zero, then there are only a finite number of $s$ s.t. $ 0 = \tilde{\target}_{i,1}(\point)$. Each of these can be checked explicitly.

We say that an eigenvalue $\lambda\in\algebraicfunctions$ (possibly constant) is a \emph{generalised root of unity} if there exists an $a\in\mathbb{N}_{\ge 1}$, such that $\lambda^{a}(\param)$ is a real-valued and non-negative function. Let $\operatorname{order}(\lambda)$ of a generalised root of unity $\lambda$ be the minimal such $a$. Notice that any real function is a generalised root of unity with order at most $2$.
When we say an eigenvalue \textit{is} a root of unity, then the eigenvalue is necessarily a constant function.

\begin{lemma}\label{lemma:nonzeroonly1}
To decide \cref{problem:afterjordan} it suffices to assume that no $\lambda_{i}$ is identically zero and that any $\lambda_{i}$ which is a generalised root of unity is real and non-negative (in particular, the only roots of unity are exactly $1$).
\end{lemma}
\begin{proof}
If $\lambda_i =0$, then $J_i^{d_i+n} = 0$ for all $n\in \mathbb{N}$, hence we only need to check $n\le d_i$ and the $s$ such that $\tilde{\target}_{i,1}(s) = \dots = \tilde{\target}_{i,d_i}(s) = 0$ (unless this holds identically, in which case the constraints from this Jordan block can be removed).
 
Take $L = \operatorname{lcm}\{\operatorname{order}(\lambda_i) \mid  \lambda_i \text{ is generalised root of unity} \}$. Then the reachability problem reduces to $L$ problems: $(J^{L})^n(J^{k}\tilde{\source}(\param)) = \tilde{\target}(\param)$ for every $k \in \{0,\dots,L-1\}$. The eigenvalue $\lambda_i^L$ corresponding to $(J_i)^{L}$ is now real and non-negative if it is a generalised root of unity.
\end{proof} %

\subsection{Jordan cells of dimension larger than \texorpdfstring{$1$}{1}}\label{jordancell:dim2}

First, we show decidability of the problem when some Jordan block has dimension at least 2:

\begin{proposition}
If there exists $J_i$  such that $d_i > 1$, then \cref{problem:afterjordan} is decidable.
\end{proposition}

There are three cases not covered by the previous section:
$\ig_i$ is not constant, $\ig_i$ is constant but not a root of unity, and $\ig_i = 1$.

Let us start with the case where $\ig_i \ne 1$, that is $\ig_i$ is a constant but not $1$, or $\ig_i$ is not a constant. Here we can use the bottom two rows from the block to obtain: \begin{equation*}
 \ig_i^n(\param)\tilde{\source}_{i,1}(\param) = \tilde{\target}_{i,1}(\param)\quad\text{ and }\quad
 \ig_i^n(\param)\tilde{\source}_{i,2}(\param) + n \ig_i^{n-1}(\param) \tilde{\source}_{i,1}(\param) = \tilde{\target}_{i,2}(\param),
 \end{equation*} 
We reformulate these equations, defining algebraic function $\theta$:
\begin{equation*}
 \ig_i^n(\param) = \gamma_i(\param) = \tilde{\target}_{i,1}(\param)/\tilde{\source}_{i,1}(\param)\quad\text{ and }\quad
  n  = \theta(\param) = \ig_i(\param)(\tilde{\target}_{i,2}(\param)/\tilde{\target}_{i,1}(\param) - \tilde{\source}_{i,2}(\param)/\tilde{\source}_{i,1}(\param))
 \end{equation*} 
Any roots or poles of $\tilde{\source}_{i,1},\tilde{\source}_{i,2},\tilde{\target}_{i,1},\tilde{\target}_{i,2},\ig_i$ can be handled manually (and we already ensured $\tilde{\source}_{i,1}$ is not identically zero). We can then apply the following lemma.

\begin{restatable}{lemma}{lemmaheightsboundnjordand} \label{lemma:thetalambdagamma}
Given algebraic functions $\ig,\gamma,\theta$ in parameter $\param$, with $\ig$ not a root of unity, then there is a bound on $n\in\naturals$ such that there exists an $\point \in \algebraics$ with 
$n = \theta(\point)$ and $\ig^n(\point) = \gamma(\point)$.
\end{restatable}
\begin{proof}[Proof sketch]
We sketch the case where $\ig$ is not a constant function, a similar (but distinct) approach is used for $\ig $ constant. Taking heights on $\ig^n(\point) = \gamma(\point)$ we obtain $nh(\lambda(\point)) = h(\gamma(\point))$, applying \autoref{lem:1} twice (on both $\lambda$ and $\gamma$) we obtain $nh(\point) = \Theta(h(\point))$. In particular if $n$ is large (say $n > A$) then $h(\point)$ is bounded (say $h(\point)< B$).
Taking heights on $n = \theta(\point)$ we obtain $\log(n) = h(n) = h(\theta(\point)) = \Theta(h(\point))$. If $n > A$ then $\log(n) \le BC$. Hence $n \le \max\{A, exp(BC)\}$.
\end{proof}

The remaining case where $\ig_i = 1$ results only in an equation of the form $n = \theta(\point)$, so $\lambda_j^n(s) = \gamma_j(s)$ can be taken from any other Jordan block where $\lambda_j \ne 1$ and again we apply \autoref{lemma:thetalambdagamma} to place a bound on $n$.

\subsection{Further simplifying assumptions for diagonal matrices}\label{sec:formulatereadyforlemmas}

Henceforth, we may assume that $J$ is a diagonal matrix resulting in the formulation of \cref{lemma:main,lemma:rank1case}: given eigenvalues $\ig_1,\ldots,\ig_t$  and so we want to know if there exists $(n,\point) \in \N \times \R\setminus \exceptionalpoints{}$ such that
\begin{equation}\label{eq:diagonalsystem}
\ig_i^n(\point)= \tg_i(\point)\qquad {\text{\rm for~all}} \qquad i=1,\ldots,t
\end{equation}
Finally we make some simplifications in \autoref{lemma:cleanuplemma}:

\begin{lemma}\label{lemma:cleanuplemma}
To decide \autoref{problem:afterjordan},  it suffices to decide the problem with instances  where the eigenvalues $\ig_i$ are distinct, that none of the $\ig_i$'s are identically zero, that none of the constant $\ig_i$'s are roots of unity, and every constant $\ig_i$ is associated with non-constant $\tg_i$.
\end{lemma}

\begin{proof}
  Consider first the case that $\lambda_1 = \lambda_2$.
If also $\gamma_1=\gamma_2$ then the equations $\lambda_1^n=\gamma_1$
and $\lambda_2^n=\gamma_2$ are equivalent and one of them can be
removed.
Otherwise, if $\gamma_1 \neq \gamma_2$,
the equations $\lambda_1^n=\gamma_1$
and $\lambda_2^n=\gamma_2$ can only have a common solution for $s\in
\mathbb{R}$ with $\gamma_1(s)=\gamma_2(s)$, i.e., we can restrict to a
finite set of parameters, in which case the problem becomes decidable.

We have already established, in \cref{lemma:nonzeroonly1}, that none of the $\ig_i$'s are identically zero, and that the only constant root of unity is $1$. Indeed if $\lambda_j = 1$ then we have $1^n =  \gamma_j(s)$, which holds either at finitely many $s$ or $\gamma_j$ is the constant $1$ and the constraint can be dropped.

If there exists $i$ with constant $\ig_i$ (not a root of unity) and constant $\tg_i$ then there is at most a single $n$ such that $\ig_i^n = \tg_i$. This $n$ can be found using the  Kannan-Lipton problem on the single constraint. The remaining constraints can be verified for this $n$ using \autoref{lemma:fixedn} to determine if they are simultaneously satisfiable.
\end{proof}

\subsection{Multiplicative dependencies}\label{sec:multiplicativedependencies}
To handle cases when the eigenvalues $\ig_i$'s are multiplicatively dependent,
we often argue as in the following manner.
Say $\ig_1^{a_1} = \ig_2^{a_2}\cdots \ig_t^{a_t}$ with $a_1 \neq 0$. Consider the system
\begin{equation}\label{eq:modsystem}
\ig_i^{a_i}(\point)^{ n} = \tg_i^{a_i}(\point)\qquad {\text{\rm for~all}} \qquad i = 1,\ldots,t.
\end{equation}
It is clear that the set $E$ of solutions $(n,\point)$ to \eqref{eq:diagonalsystem}
is a subset of the set $E'$ of solutions to \eqref{eq:modsystem}. Furthermore, for $(n,\point) \in E'$ we have
$\tg_1^{a_1}(\point) = \ig_1^{a_1n}(\point) = (\ig_2^{a_2}\cdots \ig_t^{a_t})^n(\point)
= \tg_2^{a_2}\cdots \tg_t^{a_t}(\point)$.

We conclude that if $\tg_1^{a_1} \neq \tg_2^{a_2}\cdots \tg_t^{a_t}$,
then there can only be finitely many $s$ solving \eqref{eq:modsystem}, and thus the original problem, and so the problem becomes decidable. In case
$\gamma_1^{a_1} = \tg_2^{a_2}\cdots \tg_t^{a_t}$,
the first equation in \eqref{eq:modsystem} is redundant, and we may remove it.
By repeating the process we obtain a system of the form
\eqref{eq:modsystem} where the $\ig_i$ are multiplicatively independent,
and the solutions to it contain all the solutions to the original system.

Now we face the problem of separating solutions to \eqref{eq:diagonalsystem} from the solutions to \eqref{eq:modsystem}.
If either of the sets $\{n \colon (n,\point) \in E'\}$ or
$\{s \colon (n,\point) \in E'\}$ is finite and effectively enumerable, we
can clearly decide whether $E$ is empty or not, utilising either Kannan--Lipton or
\autoref{lemma:fixedn} finitely many times. This happens in the majority of cases.
In the case that both the above sets are unbounded, we bound the suitable $n$ in case $\rank\{\ig_1,\ldots,\ig_t\} \geq 2$ in \Cref{sec:rank2}. For the case of
$\rank\{\ig_1,\ldots,\ig_t\} \leq 1$ we give a
separate argument in \Cref{sec:rank1case}.

\section{The case of \texorpdfstring{$\rank\langle \ig_1,\ldots,\ig_t\rangle \ge 2$}{the rank of eigenvalues is at least 2}}

\label{sec:rank2}

In this section we recall and prove the following \cref{lemma:main}:\mainlemma*  

By
\autoref{lemma:cleanuplemma} we may assume that none of $\ig_i$'s are
identically zero or a root of unity.

\subsection{All \texorpdfstring{$\ig_i$'s}{lambdas} constant}

In this section we sketch the proof for the case where $\ig_i$'s are all constant. We reduce to a special case of the Skolem problem, but show  that this particular instance is decidable. Since $\rank \ge 2$, we have at least two constraints and so there are constants $\ig_1$ and $\ig_2$, not roots of unity, and multiplicatively independent, with $\tg_1,\tg_2$ not constant.

\begin{restatable}{lemma}{twoeigenconstantgammanot}\label{lemma:twoeigenconstantgammanot}
Suppose $\ig_1$, $\ig_2$ are constant, not roots of unity,  multiplicatively independent, and that $\gamma_1,\gamma_2$ are non-constant functions. Then
the system $\ig_1^n = \tg_1(s)$, $\ig_2^n = \tg_2(s)$ has only finitely many solutions.
\end{restatable}
\begin{proof}[Proof Sketch]

Let the minimal polynomials over $\algebraics{}[x,y]$ of $\tg_1$ and $\tg_2$ be $P_1$ and $P_2$ with
$P_i\in \overline{\Q}[x,y_i]$. The polynomials $P_1$ and $P_2$ have no common factors as elements of $\overline{\Q}[x,y_1,y_2]$. Eliminating $x$ from these polynomials we get a non-zero polynomial $P \in \overline{\Q}[y_1,y_2]$ for which $P(\alpha_1,\alpha_2)=0$ for all $\alpha_1 = \tg_1(s)$ and $\alpha_2 = \tg_2(s)$, $s \in U$. 
The sequence
$(u_n )_{n=0}^{\infty}$, with
\[ u_n = P(\ig_1^n,\ig_2^n) = \sum_{k,\ell} a_{k,\ell} (\ig_1^{k}\ig_2^{\ell})^n,\]
$a_{k,\ell} \in \overline{\mathbb{Q}}$, is a linear recurrence sequence
over $\algebraics$, and we wish to characterise those $n$ for which
$u_n = 0$.
By the famous Skolem--Mahler--Lech theorem (see, e.g., \cite{cassels1986local}), the set of such $n$
is the union of a finite set and finitely many arithmetic progressions. Furthermore, it is decidable whether such a sequence admits infinitely many
elements, and all the arithmetic progressions can be effectively constructed \cite{BerstelM1976deux}. But, in general, the elements of the finite set are not known
to be effectively enumerable---solving the Skolem problem for arbitrary LRS
essentially reduces to checking whether this finite set is empty. However, the case at hand can be handled using now standard techniques involving powerful
results from transcendental number theory, such as Baker's theorem for
linear forms in logarithms, and similar results on
linear forms in $p$-adic logarithms (see, e.g., \cite{MignotteST84,Vereshchagin85}). We show there exists an effectively computable $n_0 \in \N$ such that $u_n \neq 0$ for all $n\geq n_0$. We give a brief sketch (a detailed proof appears in~\cref{subsec:constantsRank2}):

Assuming first that $|\ig_1|$ and $|\ig_2|$
are multiplicatively independent, it is evident
that the modulus of $u_n$ grows as $c\alpha^n + o(\alpha^n)$ for some $c\in \R_+$, where $\alpha$
is the maximal modulus of the terms
$\ig_1^k\ig_2^{\ell}$ (there is only one term with this modulus). One can straightforwardly compute an upper bound on any $n$ for which $u_n = 0$.

If the values $|\ig_1|$ and $|\ig_2|$ are
multiplicatively dependent but neither is of modulus $1$, we may again use an asymptotic argument. For this, we need Baker's theorem on linear forms in logarithms to show that a (related) sequence grows in modulus as
$c\alpha^n/n^D + o(\beta^n)$, with $\beta < \alpha$ and effectively computable constants $c$, $D$. On the other hand, if $|\ig_i|=1$ but
$\ig_1$ is an algebraic integer (a root of a monic polynomial with coefficients in $\Z$), then it will have a Galois conjugate (roots of the minimal polynomial of $\ig_1$) $\tilde{\ig_1}$ with $|\tilde{\ig_1}| > 1$. Hence a suitable Galois conjugate of the sequence $(u_n)$ will be of the form considered in the previous case, and the zeros of $(u_n)$ and $(\tilde{u}_n)$ coincide. The asymptotic argument
can be applied to $(\tilde{u}_n)$.

The final case is when $\ig_1$ and $\ig_2$
are not algebraic integers. We turn to the theory of prime ideal decompositions of the numbers $\ig$ and argue, employing a version of Baker's theorem for $p$-adic valuations (see, e.g., \cite{Vereshchagin85}) to conclude similarly that the $n$
for which $u_n = 0$ are effectively bounded above.
\end{proof}

\subsection{At least one non-constant}\label{subsec:non-constantRank2}

Henceforth, we can assume that at least one
$\ig_i$ is non-constant. We may take the $\ig_i$'s to
be multiplicatively independent with $t\geq 2$, otherwise consider a
multiplicatively independent subset of the functions: it always has at
least two elements by the assumption on $\rank$, and, furthermore, at least
one of them is not constant. The removal of equations will be done as
described in \cref{sec:multiplicativedependencies}; here we show that there
are only finitely many $n$ giving solutions $(n,s)$ to the reduced system, so we need not worry about creating too many new solutions.

The following theorems are the main technical results from the literature
utilised in the arguments that follow, formulated in a way to suit our needs. Here $\mathcal{C}(\algebraics)$ denotes the set of algebraic points in $\algebraics{}^d$ on an algebraic set $\mathcal{C} \subseteq  \C^d$.
\begin{theorem}[{\cite[Theorem~2]{BMZ1}}]\label{thm:BMZfinManyPoints}
Let $\mathcal{C}$ be an \emph{absolutely
irreducible} (irreducible in $\algebraics{(x)}$) curve defined over $\overline{\Q}$ in $\C^d$.
Assume that the coordinates of the curve are multiplicatively independent modulo constants (i.e., the points $(x_1,\ldots,x_d)\in \mathcal{C}(\overline{\Q})$ do not satisfy
$x_1^{a_1}\cdots x_d^{a_d} = c$ identically for any $(a_1,\ldots,a_d)\in \Z^d \setminus \vec{0}$, $c\in \overline{\Q}$). Then the points
$(x_1,\ldots,x_d) \in \mathcal{C}(\overline{\Q})$ for which
$x_1$, \ldots, $x_d$ satisfy at least
two independent multiplicative relations form a finite set.
\end{theorem}
We note that given the curve $\mathcal{C}$,
the finite set of points $(x_1,\ldots,x_d)$ on
$\mathcal{C}$ for which $x_1,\ldots,x_d$, satisfy at least two independent multiplicative relations
can be effectively constructed.
Indeed, this is explicitly mentioned in the last
paragraph of the introduction of \cite{BMZ1}:
the proof goes by showing effective bounds on the degree and height of such points.

\autoref{thm:BMZfinManyPoints} holds for curves in $\C^d$ for arbitrary $d$. If one allows the coordinates on the curve
to satisfy a non-trivial multiplicative relation, then there can be infinitely many such points \cite{BMZ1}. On the other hand, in
\cite{BMZ2} Bombieri, Masser, and Zannier consider relaxing the assumption of multiplicative independence modulo constants to
multiplicative independence and conjecture that
the conclusion of the above theorem still holds \cite[Conj.~A]{BMZ2}. Supporting the conjecture, \cite{BMZ2}
proves a theorem which will suffice for us.

\begin{theorem}\label{thm:BMZ2finManyPoints}
Let $\mathcal{C}$ be an absolutely irreducible curve in $\C^d$ defined over $\overline{\Q}$. Assume that the the coordinates of the curve are multiplicatively independent, but $\mathcal{C}$ is contained in a set
of the form $\vec{b}H$, where $H$ is the set of points in $\algebraics{}^d$ satisfying at least $d-3$ independent multiplicative relations%
\footnote{With $b = (b_1,\ldots,b_k)$, here $\vec{b}H = 
\{(b_1 x_1,\ldots,b_d x_d) \colon (x_1,\ldots,x_k)\in H
\}$ is a coset of a \emph{subgroup of dimension at most $3$} in the terminology of \cite{BMZ2}.}. Then the points
$(x_1,\ldots,x_d) \in \mathcal{C}(\overline{\Q})$ for which
$x_1$, \ldots, $x_d$ satisfy at least
two independent multiplicative relations form a finite set.
\end{theorem}
Again the finite set of points can be effectively computed.%
\footnote{In \cite{BMZ1,BMZ2} the proof is given for $d \geq 4$, and is constructive, while the case of $d=3$ is attributed to a (non-constructive) result of Liardet \cite{liardet1974conjecture}. A completely effective proof of the case can be found in \cite{Gy}.}

Let us proceed case by case.

\begin{restatable}{lemma}{dependentModCNonconstant}
\label{lem:dependentModCNonconstant}
Assume that $\{\ig_1,\ldots,\ig_t\}$ is multiplicatively dependent modulo constants, but is multiplicatively independent. Then there exists a computable constant
$n_0$ such that system \eqref{eq:mainlemmaequations} admits no solutions for $n > n_0$.
\end{restatable}

We may now focus on sets $\{\ig_1,\ldots,\ig_t\}$ that are multiplicatively independent modulo constants. We still might have multiplicative dependencies between the $\ig_i$ and $\tg_i$. We take care of these cases in the remainder of this section.

\begin{lemma}\label{lem:multiplicativelyindependetTargetsAndEigenvalues}
Assume that $\{\ig_1, \ig_2, \tg_1, \tg_2\}$ is multiplicatively independent. Then system \eqref{eq:mainlemmaequations} admits only finitely many solutions, all of which can be effectively enumerated.
\end{lemma}
\begin{proof}
We show that the set of $s$ for which the equality can hold is finite and
such $s$ can be computed. We employ the powerful Theorems \ref{thm:BMZfinManyPoints} and \ref{thm:BMZ2finManyPoints} of Bombieri, Masser, and Zannier, from which the claim is immediate. We first prime the situation as follows.

Let that $\ig_1$, $\ig_2$, $\tg_1$, $\tg_2$ have minimal polynomials
$P_1 \in \Q[x,x_1]$, $P_2 \in \Q[x,x_2]$, $P_3 \in \Q[x,x_3]$,
$P_4 \in  \Q[x,x_4]$, respectively. Eliminating $x$ from $P_1$ and $P_2$ (resp., $P_3$, $P_4$), we get a
polynomial $Q_1 \in \Q[x_1,x_2]$ (resp., $Q_2 \in \Q[x_1,x_3]$, $Q_3\in \Q[x_1,x_4]$) for which we have $Q_1(\ig_1(x),\ig_2(x)) = 0$ (resp., $Q_2(\ig_1(x),\tg_1(x))=0$, $Q_3(\ig_1(x),\tg_2(x))=0$) for all $x$. Let $\mathcal{C}$ be the curve defined by
$\mathcal{C}:= \{(x_1,x_2,x_3,x_4)\in \C^4 \colon Q_1(x_1,x_2) = Q_2(x_1,x_3) = Q_3(x_1,x_4) = 0\}$ and consider any of its finitely many absolutely irreducible components $\mathcal{C}'$.
We are now interested in the pairs of multiplicative relations $(n,0,-1,0)$ and $(0,n,0,-1)$ (corresponding to $x_1^n = x_3$,
$x_2^n = x_4$), for $n \geq 1$, along the curve $\mathcal{C}'$. Indeed, for any fixed
$n$, the two relations are independent in
$\algebraics^4$, i.e., neither is a consequence of the other, as they involve disjoint sets of coordinates.

First assume that $\ig_1,\ig_2,\tg_1,\tg_2$ are multiplicatively independent modulo constants. Then so are the points on the curve
$\mathcal{C}'$, and the result follows from \autoref{thm:BMZfinManyPoints}
as the result is constructive.

Otherwise $\ig_1,\ig_2,\tg_1,\tg_2$ are multiplicatively
dependent modulo constants but are multiplicatively independent. Then $\mathcal{C}'$ is contained in a set of the form $\vec{b}H$,
where $H$ satisfies at least one multiplicative relation. Applying \autoref{thm:BMZ2finManyPoints}
with $d = 4$, the points on $\mathcal{C}'$
satisfying $x_1^n = x_3$ and $x_2^n = x_4$ for any $n\geq 1$, form an effectively constructable finite set.
\end{proof}

To complete the proof of \autoref{lemma:main}, we need to show
the claim holds when $\ig_1,\ig_2,\tg_1,\tg_2$ are
multiplicatively dependent, while $\ig_1$ and $\ig_2$ are multiplicatively
independent modulo constants. The proof goes along the same lines as in the above with some extra technicalities. 

\begin{restatable}{lemma}{lemmultdepindepmodconst}\label{lem:multdepindepmodconst}
Assume that $\ig_1$, $\ig_2$, $\tg_1$, $\tg_2$ are multiplicatively dependent, while $\ig_1$, $\ig_2$ are multiplicatively independent modulo constants. Then there exists a computable constant
$n_0$ such that system \eqref{eq:mainlemmaequations} admits no solutions for $n > n_0$.
\end{restatable} %

\section{The case of \texorpdfstring{$\rank\langle \ig_1,\ldots,\ig_t\rangle =1$}{the rank of eigenvalues is 1}}
\label{sec:rank1case}
This section recalls and sketches the proof of \cref{lemma:rank1case}.

\rankonelemma*

As sketched in \cref{sec:multiplicativedependencies}, since there is a multiplicative dependence between functions, we first show that, without loss of generality, there is a single equation $\lambda^n(s) = \gamma(s)$. %
\begin{restatable}{lemma}{wlogsingle}
Suppose $\rank\langle \ig_1,\ldots,\ig_t\rangle =1$, then whether there is a solution $(n,\point)\in {\mathbb N}\times\mathbb{R}\setminus \exceptionalpoints$ to $\ig_i^n(\point)=\tg_i(\point)$ for all $i=1,\ldots,t$ reduces to instances with $t=1$.
\end{restatable}

We then separate into the case where $\lambda$ is real and the case where $\lambda$ is complex. Let us start by assuming $\lambda$ is a real function.
\begin{restatable}{lemma}{rankonereallemma}
Given real algebraic functions $\lambda$ and $\gamma$, it is decidable whether there exists $(n,\point)\in\N\times \mathbb{R}\setminus\exceptionalpoints{}$  such that $\lambda^n(\point) = \gamma(s)$.
\end{restatable}
\begin{proof}[Proof Sketch]
The interesting case occurs on an interval $S=(s_0,s_1)$ on which $0< \lambda(s),\gamma(s)< 1$ for $s\in S$. Other cases either reduce to this case, or occur for finitely many $s$ which can be checked independently. The function $\gamma(s)$ is fixed between $s_0,s_1$. Each point $\lambda(s)^n$ decreases with every $n$. One can test for each $n$ whether the lines $\lambda(s)$ and $\gamma(s)$ intersect, or one can find some bound $n_0$ after which $\lambda(s)^n < \gamma(s)$ for all $s\in S$ and $n > n_0$, so one can be sure there is no solution.
\end{proof}

Secondly, we consider the case $\ig$ takes on complex values. In this case, since $\ig_i$ was a complex eigenvalue of $M$, then so too is its conjugate $\overline{\ig_i}$, yet $\ig_i$ and $\overline{\ig_i}$ are multiplicatively dependent, in which case it turns out that $|\ig| = 1$.

\begin{restatable}{lemma}{rankOneComplex}
\label{lem:rankOneComplex}
Let $\ig$ and $\tg$ be algebraic functions. Assume $\lambda$ is not real, non-zero, not a root of unity, and of modulus 1. The equation $\ig(s)^n = \tg(s)$ admits solutions as follows. If $\tg$ is not of modulus $1$ constantly, then there are finitely many $s$. If $\tg$ is of modulus 1 identically and $\ig$ is constant, then there are infinitely many solutions and such a solution can be effectively found. Finally, if $\ig$ is not constant, then the equation admits a solution for all $n \geq n_0$, and $n_0$ is computable.
\end{restatable}
\begin{proof}[Proof Sketch]
The interesting case turns outs to be when $\lambda$ and $\tg$ both define arcs on a unit circle. By taking powers of $\lambda$ the arc grows, and eventually encompasses the arc defined by $\tg$. The intermediate value theorem then implies there is an $s$ satisfying $\lambda^n(s) = \gamma(s)$.
\end{proof}

\appendix

\section{Additional Material for Section \ref{sec:prelim}}\label{app:prelim}

\multiplicative*

\begin{proof}[Proof sketch]
An algebraic function $\ig$ can be
  expressed as a converging \emph{Puiseux series}
  $\ig(x) = \sum_{n=n_0}^{\infty}c_n(x-\alpha)^{n/\deg_y(P_{\ig})}$
  for some $c_n \in \overline{\Q}$, $c_{n_0} \neq 0$, and $n_0 \in \Z$,
  around a point $\alpha \in \overline{\Q}$ (see
  \cite[Chapt.~1.8.13]{Foster81compact}, \cite{KungT1978Algebraic}). Evidently any
  $\alpha \in \overline{\Q}$ has \emph{order}
  $\ord_{\ig}(\alpha) = n_0/\deg_y(P_{\ig})\in \Q$, i.e., the exponent of the first term in the Puiseux series. Let $\alpha_1$,
  \ldots, $\alpha_k$ be the roots and poles of the $\ig_i$. With each
  $\ig_i$ we associate the vector
  $g_i = (\ord_{\ig_i}(\alpha_j))_{j=1}^k$. Now
  $\sum_{i=1}^k a_i g_i = \vec{0}$ implies that the function
  $\ig_1^{a_1}\cdots \ig_t^{a_t}$ has no roots or poles, hence is
  constant by Bezout's theorem. Compute a basis for
  $\ker((g_{i,j})_{i=1,j=1}^{t,k})\cap \Z^k$
  (\cite[Cor.~5.3c]{Schrijver1999theory}) for the claim on $L'$. As $L'$ is finitely generated, $L$
  can be seen as the set of multiplicative relations of a finite
  set of algebraic numbers. A generating set for $L$ can be found utilising a
  deep result of Masser \cite{Masser1988Linear}
  (\cite{cai2000complexity,Ge1993algorithms}).
\end{proof}

\boundhlsbyhs*
\begin{proof}
To this end,
let $P \in \overline{\Q}[x,y]$ be irreducible, with $d_x$ the maximal degree of $x$, and $d_y$ that of $y$, and assume $d_x,d_y \geq 1$. Let
$P(\alpha,\beta) = 0$ with algebraic $\alpha$, $\beta$. It is known that there exists a constant $C_P$
depending on $P$ such that
\begin{equation*}
\left| \frac{h(\alpha)}{d_y} - \frac{h(\beta)}{d_x} \right| \leq  C_P \sqrt{\max\left\{ \frac{h(\alpha)}{d_y} , \frac{h(\beta)}{d_x}\right\}}.
\end{equation*}
For example, the main result of \cite{habegger2017quasi} shows that
\[C_P = 5\left(\log\left(2^{\min\{d_x,d_y\}}(d_x+1)(d_y+1)\right) + h_p(P)\right)^{1/2}\] suffices%
\footnote{Here $h_p(P)$
is the \emph{height of the polynomial} $P$ (see \cite[Equation~(4)]{habegger2017quasi}). For us it suffices to know that
$h_p(P)$ is at most the sum of the heights of the non-zero coefficients of
the polynomial.}, and hence an upper bound for $C_P$ is computable, given $P$. %

Now let  $P$ be the minimal polynomial of $\ig$. We have for all admissible
$s$: $P(s,\ig(s)) = 0$. Since $\ig$ is not constant, we have that the polynomial contains both $x$ and $y$, and we
may apply the above to get
\[
\left|\frac{h(s)}{d_y} - \frac{h(\ig(s))}{d_x}\right| \leq  C_P \sqrt{\max\{ h(s)/d_y , h(\ig(s))/d_x\}}.\]

By taking $c_1=d_x/(2d_y), c_3=2d_x/d_y$ and $
c_2=c_4=4C_P^2\max\{d_x,d_y\}$, we  have
\[c_1h(s) - c_2 \le  h(\ig(s)) \le  c_3h(s) + c_4.\qedhere\]
\end{proof} %

\section{Additional Material for Section \ref{sec:hardness}}
\label{appen:hardness}

In this section we show that each semialgebraic $S \subseteq \R^3$ can be written as a finite union of sets of the form $\{ v(s,t) : s,t \in \R\}$.

One way to define dimension of a semialgebraic set is using \emph{cell decomposition}. We have that a semialgebraic set $S \subseteq \R^3$ of dimension 2 can be written as a finite union of 2-cells and 1-cells in $\R^3$, and that a $d$-cell is semialgebraically homeomorphic to the open hypercube $(0,1)^d$ \cite{bochnak2013real}. Hence to show our main result it suffices to show how to write $C \subseteq \R^3$, $C = f((0,1)^d)$ where $f$ is a semialgebraic function, as a union of sets parametrised using two parameters and algebraic functions in two variables.

First, let us consider parametrisation of very simple sets in $\R$. Observe that a point $p \in \algebraics$ can be characterised using the algebraic function $f(s) = p$, the interval $(0,1]$ as $\{\frac{1}{1+s^2} \colon s \in \R\}$ and the interval $(0, \infty)$ as $\{\frac{1}{s^2} \colon s \in \R\}$. We can characterise other intervals using these characterisations. For example, $(a,b]=\{a + \frac{b-a}{1 + s^2}\}$, $[b, a) = \{a-\frac{a-b}{1+s^2}: s \in \R\}$ and an open interval $(a,b)$ can be written as $(a, b) = (a, \frac{a+b}{2}] \cup [\frac{a+b}{2}, b)$. 

Next, a couple of useful lemmas.
\begin{lemma}
\label{trproof}
Let $g \colon \R \mapsto \R$ be a semialgebraic function. The graph $G = \{ (x, g(x) \colon x \in \R\}$ can be written as a union of sets of the form $\{ v(s) \colon  s \in \R\}$.
\end{lemma}
\begin{proof}
By definition, the function $g$ is semialgebraic if and only if its graph $G$ is a semialgebraic subset of $\R^2$. Let $p_1(x,y) = 0, q_1(x,y)>0, \ldots , q_m(x,y)>0$ be the constraints that define $G$ (recall that one can define a semialgebraic set using only one equality constraint). Viewing $p_1, q_1, \ldots, q_m$ as polynomials in $y$, we can factorise
\[
\begin{cases}
	p_1(x, y) = (y - h^0_1(x)) \cdot  \ldots \cdot (y - h^0_{\kappa(0)}(x)) = 0\\
	q_1(x, y) = (y - h^1_1(x)) \cdot  \ldots \cdot (y - h^1_{\kappa(1)}(x)) > 0\\
	\cdots\\
	q_m(x, y) = (y - h^m_1(x)) \cdot  \ldots \cdot (y - h^m_{\kappa(m)}(x)) > 0
\end{cases}
\]
where $h^i_r$ is an algebraic function for every $0 \leq i \leq m$ and $1 \leq r \leq \kappa(i)$. Next we will show how to compute $\kappa(0)$ subsets $I_1, \ldots, I_{\kappa(0)}$ of $\R$ that have the following properties.
\begin{enumerate}
	\item $\bigcup_{j=1}^{\kappa(0)}I_j = \R$;
	\item Each $I_j$ is a finite union of intervals;
	\item For $1 \leq j \leq \kappa(0)$, the value of $y$ for each $x \in I_j$ is equal to $h^0_j(x)$, the $j$th root of $p_1$.
\end{enumerate}
This will allow us to write
\[
G= \bigcup_{j=1}^{\kappa(0)} \{(x, h^0_j(x)) : x \in I_j\}.
\]
Recall that each $I_j$ is a finite union of intervals, each of which can be parametrised by an algebraic function with domain $\R$. Since composition of two algebraic functions remains algebraic, we can characterise each component of $G$ that comes from a single subinterval of $I_j$ using an algebraic function with domain $\R$. Hence we can write $G$ as a union of sets with the desired parametrization.

To construct $I_j$, we proceed as follows. From Condition~3 above, $I_j = \{x : (x, h^0_j(x)) \in G\}$ and hence can be defined by the formula
\[
\varphi(s) = p_1(x, h^0_j(x)) = 0 \land q_1(x, h^0_j(x)) > 0 \land \cdots \land q_m(x, h^0_j(x)) > 0.  
\]
Hence $I_j$ is semialgebraic. Since semialgebraic sets have finitely many connected components, $I_j$ must be a finite union of interval subsets of $\R$. 
\end{proof}

\begin{lemma}
\label{trproof2}
	Let $D \subseteq \R^2$ be semialgebraic. $D$ can be written as
	\[
	D = \bigcup_{i=1}^k D_i = \bigcup_{i=1}^k \{ v_i(s,t) : s,t \in \R\}
	\]
	where for each $i$, $v_i$ is algebraic over $\Q(s, t)$.
\end{lemma}
\begin{proof}
	By cell decomposition, $D$ must be a union of 
	\begin{enumerate}
		\item points,
		\item sets of the form $\{(x, g(x)) : x \in (a,b)\}$ where $g : \R \rightarrow \R$ is semialgebraic, and
		\item sets of the form $\{(x,y) :  x \in (a,b) , g(x) < y < h(x)\}$ where $g,h$ are semialgebraic.
	\end{enumerate}

Sets of the last kind are bands between the graphs of $g$ and $h$ over the open interval $(a,b)$. We need to show that sets of each kind can be parametrized using two parameters and algebraic functions in two variables. The first two cases are handled by the preceding arguments. For the third case, let $(a,b)$, the graph of $G$ and the graph of $H$ be parametrized by the one-variable algebraic functions $v_1$, $v_2$ and $v_3$, respectively.  Then the sets of the third type can be written as $\{(v_1(s), v'(s,t)) : s, t \in \R \}$ where $v'(s,t)$ parametrizes the open interval $(g(s), h(s))$ based on the discussion above about parametrizing intervals in $\R$.
\end{proof}

Finally, we are ready to prove our main result. Let  $C \subseteq \R^3$, $C = f((0,1)^2)$ where $f$ is a semialgebraic function. Let $(u,v)$ denote a point in $(0,1)^2$ and $x(u,v), y(u,v), z(u,v)$ denote the semialgebraic functions that give us the $x,y,z$ coordinates of the point $f(u,v)$, respectively. To parametrize $C$, it suffices to parametrize the graphs of the functions $x(u,v), y(u,v), z(u,v)$.

Wlog consider $X = \{(u,v,x(u,v) : (u,v) \in (0,1)^2\}$, i.e. the graph of the function $x(u,v)$. Let $p_1(x_1, x_2, x_3) = 0, q_1(x_1, x_2, x_3) > 0, \ldots, 
q_m(x_1, x_2, x_3) > 0$ be the constraints defining $X$. We proceed in the same way as in the proof of Lemma~\ref{trproof}. Vieweing $p_1, q_1, \ldots, q_m$ as polynomials in $x_3$, we factorize to obtain
\[
\begin{cases}
	p_1(x_1, x_2, x_3) = (x_3 - h^0_1(x_1, x_2)) \cdot  \ldots \cdot (x_3 - h^0_{\kappa(0)}(x_1, x_2)) = 0\\
	q_1(x_1, x_2, x_3) = (x_3 - h^1_1(x_1, x_2)) \cdot  \ldots \cdot (x_3 - h^1_{\kappa(1)}(x_1, x_2)) > 0\\
	\cdots\\
	q_m(x_1, x_2, x_3) = (x_3 - h^m_1(x_1, x_2)) \cdot  \ldots \cdot (x_3 - h^m_{\kappa(m)}(x_1, x_2)) > 0
\end{cases}
\]
where each $h^i_r$ is algebraic over $\Q(x_1,x_2)$. We then compute $\kappa(0)$ semialgebraic subsets $S_1, \ldots, S_{\kappa(0)}$ of $\R^2$ that have the following properties.
\begin{enumerate}
	\item $\bigcup_{j=1}^{\kappa(0)}S_j = (0,1)^2$;
	\item For $1 \leq j \leq \kappa(0)$, the value of $x_3$ for each $(x_1, x_2) \in S_j$ is equal to $h^0_j(x_1, x_2)$, the $j$th root of $p_1$.
\end{enumerate}
This will allow us to write
\[
X= \bigcup_{j=1}^{\kappa(0)} \{(x_1, x_2, h^0_j(x_1, x_2)) : (x_1, x_2) \in S_j\}.
\]
Now it only remains to observe that the unit square and, by Lemma~\ref{trproof2}, each $S_j$ can be parametrized using two parameters and algebraic functions. %
\section{Additional Material for Section \ref{jordancell:dim2}}

\lemmaheightsboundnjordand*
\begin{proof}

If $\theta(\param)$ is constant, then $n$ is uniquely determined. If not, by applying heights we get that  $\log(n) =h(n) =  h(\theta(\param))$ and by \cref{lem:1} we get  $a_1,a_2,a_3,a_4>0$ such that
\begin{equation}\label{eq:bound:logn}
a_1 h(\param) -a_2 \le h(\theta(\param)) = \log(n) \le a_2h(\param) + a_3.
\end{equation}

Now we split into the cases where $\ig$ is constant or not.

If $\ig$ is constant, then there exists fixed  $b = h(\ig)$, such that  $h(\ig^n) =  nh(\ig) = b n$.

Requiring that $\ig^n = \gamma(\param)$   and using  \cref{lem:1} on the algebraic function $\gamma(\param)$ we obtain $c_3,c_4$ such that
\begin{equation}\label{eq:bnbyh}
bn = h(\ig^n) =  h(\gamma(\param)) \le c_3h(\param) + c_4
\end{equation}
Combining \cref{eq:bnbyh} and \cref{eq:bound:logn} we obtain
 \[bn \le  c_3h(\param)+c_4\le c_3(\log(n) +a_2)/a_1 +c_4,\]
 which implies:
\begin{align*}
 \sqrt{n} &\le  \frac{n}{\sqrt{n}}\le \frac{n}{\log(n)} 
 \\ &\le
\frac{1}{b}\left[\frac{c_3}{a_1} + \frac{c_3a_2/a_1 + c_4}{\log(n)}\right]
\\ &\le
\frac{c_3 + c_3a_2}{ba_1} + \frac{c_4}{b} \qquad \text{if }n \ge 3.
\end{align*}
Thus we bound $n$:\[n \le \max\left\{3,\left(\frac{c_3 + c_3a_2}{ba_1} +  \frac{c_4}{b}\right)^2\right\}.
\]

We now consider $\ig(\param)$ not a constant function. Then from \cref{lem:1} we obtain $b_1,b_2,c_3,c_4$ such that  
\[b_1h(\param) - b_2 \le h(\ig(\param)) \text{ and } h(\gamma(\param)) \le c_3 h(\param) + c_4 \]
Using $nh(\ig(\param)) =  h(\ig^n(\param))=  h(\gamma(\param))$ we obtain $n(b_1h(\param)- b_2) \le c_3 h(\param) + c_4 $ which bounds $h(\param)$:
\[
h(\param)\le \frac{nb_2 + c_4}{nb_1-c_3} \le \frac{2b_2+2c_4}{b_1}  \text{ if } n \ge \max\left\{\frac{2c_3}{b_1}, 1\right\}.
\]
Finally we bound $n$ using \cref{eq:bound:logn}:  
\[ \log(n)\le a_3h(\param)  + a_4 \le a_3\left(\frac{2b_2+2c_4}{b_1}\right) +a_4.\]Taken together we have \[n \le\max\left\{\frac{2c_3}{b_1}, 1, \exp\left(\frac{a_3(2b_2+2c_4)+a_4b_1}{b_1} \right)\right\}.\qedhere\]
\end{proof}

Let us now deal with the case where $d_i > 1$ and $\ig_i= 1$. The equations formed by the constraints of 
$\left(\begin{smallmatrix}
		1 & 1\\
	 &  &\ddots &1
	 \\ & & &1
	 \end{smallmatrix}
	 \right)^n\left(\begin{smallmatrix}
		\tilde{\source}_{i,d_i}\\
		\vdots \\
		\tilde{\source}_{i,1}
	 \end{smallmatrix}
	 \right)=  \left(\begin{smallmatrix}
		\tilde{\target}_{i,d_i}\\
		\vdots \\
		\tilde{\target}_{i,1}
	 \end{smallmatrix}
	 \right)$ 
 describe a set of polynomial equations in variable $n$ and coefficients in $\mathbb{K}$:
\begin{equation*}
\Big\{\tilde{\source}_{i,1}(\param) 	= \tilde{\target}_{i,1}(\param), \quad
\tilde{\source}_{i,2}(\param) + n \tilde{\source}_{i,1}(\param) = \tilde{\target}_{i,2}(\param), \quad\ldots,\quad 
\sum_{i=1}^k \binom{n}{i}\tilde{\source}_{i,i}(\param) = \tilde{\target}_{i,k}(\param)\Big\}.
\end{equation*}
Let us consider all such equations formed by $J_i$ such that $\lambda_i =1$. Clearly $\tilde{\source}_{i,1}= \tilde{\target}_{i,1}$ identically, or else there are finitely many $s$ such that $\tilde{\source}_{i,1}(\point)= \tilde{\target}_{i,1}(\point)$. Hence, the first equation can essentially be dropped. Using the second equation to replace $n$ by $(\tilde{\target}_{i,2}- \tilde{\source}_{i,2}) /\tilde{\source}_{i,1}$ in all other such equations gives a collection algebraic function only in $\param$. These functions are either identically zero, or have finitely many solutions. If any one function has finitely many instantiations of $x$ then we only need to check these instantiations.

If all of the resulting functions are identically zero, then the system of equations is equivalent to the single equation $n = \theta(\param)$, where  $\theta(\param) = \frac{\tilde{\target}_{i,2}(\param)- \tilde{\source}_{i,2}(\param)}{\tilde{\source}_{i,1}(\param)}$. We can first verify whether the range of $\theta(\param)$ over $\param$ is bounded.  If it is, test every integer $n$ in the range (by \autoref{lemma:fixedn}).

In the remaining case, $\theta(\param)$ is unbounded, so there is a solution to $n= \theta(\param)$ for every large $n$. If this is the only equation, we are done (and the answer is \textsc{yes}). Alternatively there is some other constraint, which we can take from the bottom row of some different Jordan block: $\ig_j(\param)^n = \gamma_j(\param)$. We can assume  $\ig_j$ not a root of unity because the only root of unity was $1$, for which all of the constrains are encoded in $n = \theta(\point)$. We can now apply the following lemma, which places a bound on $n$ when $n$ appears both linearly and as an exponent w.r.t. algebraic functions:

Again we have an instance of \autoref{lemma:thetalambdagamma} bounding $n$ that need to be checked.
\section{Additional Material for Section~\ref{sec:rank2}}\label{sec:constant}

We complete the proof of \autoref{lemma:main}:
\mainlemma*

To do this, we prove \autoref{lemma:twoeigenconstantgammanot},
and prove the remaining cases of \autoref{subsec:non-constantRank2}.

\subsection{Proof of \autoref{lemma:twoeigenconstantgammanot}}
\label{subsec:constantsRank2}
In this part we complete the proof of \autoref{lemma:twoeigenconstantgammanot}. First we recall some notions from algebraic number theory.
Most of the results appear in standard text books on the topic such as \cite{CohenBook}, but
an accessible account sufficient for our purposes can be found in \cite{Halava2005skolem}.
An \emph{algebraic integer} is an algebraic number with monic minimal polynomial in $\Z[x]$.
Let $K$ be a finite extension of $\Q$, and consider the set $\mathcal{O}_K$
of algebraic integers in $K$. The set $\mathcal{O}_K$ forms a
subring of $K$, the so-called \emph{ring of integers of $K$}. The ideals of 
$\mathcal{O}_K$ are finitely generated, and they form a commutative ring.
An ideal $P \neq [1],[0]$ (here $[\alpha]$ is the principal ideal generated by $\alpha$)  is called a \emph{prime ideal} if $P = IJ$, for some ideals $I$, $J$, implies that either $I = [1]$ or $I = P$.
Each ideal $I \neq [0]$ of $\mathcal{O}_K$ can be represented as a product of \emph{prime ideals}: $I = P_1^{k_1}\cdots P_t^{k_t}$, $k_i \geq 0$, and is unique up to the ordering of the prime ideals in the product.

For a prime ideal $P$ we define the \emph{valuation} $\nu_P \colon \mathcal{O}_K\setminus \{0\} \mapsto \N$ as follows: for
$\alpha \in \mathcal{O}_K$, $\alpha \neq 0$ and
$[\alpha] = P_1^{k_1}\cdots P_t^{k_t}$, where each $P_i$ is a prime ideal,
we set $\nu_P(\alpha) = k_i$ if $P = P_i$, and $\nu_P(\alpha) = 0$ if
$P \neq P_1,\ldots,P_t$. By convention we set $\nu_P(0) = \infty$. The valuation $\nu_P$ can be extended to
the whole number field $K$ by noting that if $\alpha$ is not an algebraic integer, then there exists $m\in \N$, $m\geq 1$, such that $m\alpha = \alpha_1$ is an algebraic integer. In this case we define $\nu_P(\alpha) = \nu_P(\alpha_1) - \nu_P(m)$, and it can be shown that this is well-defined (i.e., does not depend on the choice of $\alpha_1$ and $m$).

We need the following properties: for $\alpha$, $\beta \in K$, and $P$ a prime ideal of $\mathcal{O}_K$,
\begin{itemize}
\item $\nu_P(\alpha \beta) = \nu_P(\alpha) + \nu_P(\beta)$.
\item $\nu_P(\alpha + \beta) \geq \min\{\nu_P(\alpha), \nu_P(\beta)\}$.
\item If $\nu_P(\alpha) < \nu_P(\beta)$ then $\nu_P(\alpha + \beta) = \nu_P(\alpha)$.
\item If $\alpha \notin \mathcal{O}_K$, then there is a prime
ideal $P$ such that $\nu_P(\alpha) \neq 0$. Furthermore, such a prime ideal
can be found effectively.
\end{itemize}

We shall employ a version Baker's theorem as formulated in \cite{baker1993logarithmic}:
\begin{theorem}[Baker and Wüstholz]
Let $\alpha_1$, \ldots, $\alpha_t \in \C \setminus \{0,1\}$ be algebraic numbers different
from $0$ or $1$, and let $b_1$, \ldots, $b_t \in \Z$ be integers.
Write $\Lambda = b_1\log \alpha_1 + \ldots + b_t \log \alpha_t$,
where $\log$ is any branch of the complex logarithm function.

Let $A_1$, \ldots $A_t$, $B$ be real numbers larger than $\eu$ such that
$h(\alpha_i) \leq A_i$, and $|b_i| \leq B$ for each $i$. Let further $d$
be the degree of the extension field $\Q(\alpha_1,\ldots,\alpha_t)$ over
$\Q$.

If $\Lambda \neq 0$, then
\begin{equation*}
\log|\Lambda| > -(16 t d)^{2(t+2)} \log A_1 \cdots \log A_t \log B.
\end{equation*}
\end{theorem}

As a straightforward consequence we have the following
\begin{corollary}\label{cor:bakercor}
For algebraic numbers $\mu$ and $\zeta$ of modulus $1$ with $\mu$ not a root of unity, we have $|\mu^n - \zeta| > a/n^b$ for all large enough $n$ and
for some effectively computable constants $a>0$ and $b \in \N$ depending on $\mu$
and $\zeta$.
\end{corollary}
For a proof, see \cite[Cor.~8 of Extended Version]{OuaknineW14}.

We shall also employ a $p$-adic version of Baker's theorem proved by K. Yu \cite{yu1999p-adicI}. We employ a version which follows from a version stated in the introduction of K. Yu
 \cite{yu1999p-adicII} (for definitions, we refer to \cite{CohenBook}):
\begin{theorem}
Let $\alpha_1$, \ldots, $\alpha_t$ $(t\geq 1)$ be non-zero algebraic
numbers and $K$ be a number field containing $\alpha_1$, \ldots,
$\alpha_t$, with $d$ the degree of the extension. Let $\mathfrak{p}$ be a
prime ideal of $\mathcal{O}_K$, lying above the prime number $p$, by
$e_{\mathfrak{p}}$ the ramification index of $\mathfrak{p}$, and by
$f_{\mathfrak{p}}$ the residue class degree of $\mathfrak{p}$. For
$\alpha \in K$. Let $b_1$, \ldots, $b_t \in \Z$, and assume that
$\Xi:=\alpha_1^{b_1}\cdots \alpha_t^{b_t} - 1 \neq 0$. Let further
$h_j = \max(h(\alpha_j),\log p)$ for $j=1,\ldots,t$. Let
$B = \max\{|b_1|,\ldots, |b_t|,3\}$. Then
\begin{equation*}
\nu_\mathfrak{p}(\Xi) < 19(20\sqrt{t+1} d)^{2(t+1)} e_{\mathfrak{p}}^{t-1}\cdot
\frac{p^{f_{\mathfrak{p}}}}{(f_{\mathfrak{p}} \log p)^2}
	\log(e^5 t d) h_1\cdots h_t \log B
\end{equation*}
\end{theorem}
All the above values are effectively computable given the numbers $\alpha_1$, \ldots, $\alpha_t$.
We have a straightforward corollary:
\begin{corollary}\label{cor:p-adiclogform}
Let $\mu$ and $\zeta$ be algebraic numbers of modulus $1$ and assume $\mu$ is not a root of unity. Let $K = \Q(\mu,\zeta)$ and $\mathfrak{p}$ be a prime ideal of $\mathcal{O}_K$.
Then $\nu_{\mathfrak{p}}(\mu^n - \zeta) < C \log n$ as $n\to \infty$ for some
effectively computable constant $C$ that depends on $\mathfrak{p}$,
$\mu$ and $\zeta$.
\end{corollary}
\begin{proof}
We have $\nu_{\mathfrak{p}}(\mu^n - \zeta) = \nu_{\mathfrak{p}}(\zeta) + \nu_{\mathfrak{p}}(\mu^n\zeta^{-1} - 1)$. Since $\mu$ is not a root of unity,
 the height of $\mu^n$ increases linearly in $n$. 
\end{proof}

We now recall and prove \cref{lemma:twoeigenconstantgammanot}:
\twoeigenconstantgammanot*
\begin{proof}

Let the minimal polynomials of $\tg_1$ and $\tg_2$ be $P_1$ and $P_2$ with
$P_i \in \Q[x,y_i]$. Eliminating $x$ from these polynomials we get a
non-zero polynomial $P \in \overline{\Q}[y_1,y_2]$. For points $\alpha_1 = \tg_1(s_0)$ and $\alpha_2 = \tg_2(s)$ we have $P(\alpha_1,\alpha_2) = 0$.
We are interested in those
$n \in \mathbb{N}$ for which $P(\ig_1^n,\ig_2^n) = 0$. The sequence
$(u_n )_{n=0}^{\infty}$, with
\begin{equation}\label{eq:unAsExpPol}
 u_n = P(\ig_1^n,\ig_2^n) = \sum_{k,\ell} a_{k,\ell} (\ig_1^{k}\ig_2^{\ell})^n,
 \end{equation}
 $a_{k,\ell} \in \overline{\mathbb{Q}}$, is a linear recurrence sequence over $\algebraics$. We wish the characterise those $n$ for which $u_n = 0$.

We first consider the case that $|\ig_1|$ and $|\ig_2|$ are
multiplicatively independent, that is, $|\ig_1^{i}\ig_2^{j}| \neq 1$ for
all $i,j \in \mathbb{Z}$.
\begin{claim}
If $|\ig_1|$ and $|\ig_2|$ are multiplicatively independent, then
there exists an effectively computable $n_0 \in \N$ such that $u_n \neq 0$ for $n \geq n_0$.
\end{claim}
\begin{claimproof}
There is a unique pair $k,\ell$,
with $\ig_1^k \ig_2^{\ell}$ dominant in modulus. Then $( u_n )_n$ has a unique
dominant characteristic root, and hence there are only finitely many $n$ for
which $u_n = 0$. Indeed, $|u_n|$ grows as $|a_{k,\ell}||\ig_1^k \ig_2^{\ell}|^n + o(|\ig_1^k \ig_2^{\ell}|^{n})$, and so $u_n \neq 0$ for all $n\geq n_0$ for some $n_0$. Now $n_0$ can be clearly computed using the closed form expression \eqref{eq:unAsExpPol} of $u_n$. 
\end{claimproof}
In case the assumption of the above lemma holds, the problem
becomes decidable using \autoref{lemma:fixedn} for $n \leq n_0$.

In the remainder of this section we assume that $|\ig_1|$ and $|\ig_2|$
are multiplicatively dependent. In fact, we may assume that $|\ig_1| = |\ig_2|$: We have $|\ig_1|^i = |\ig_2|^j$ for some $i,j \in \mathbb{Z}$.
By considering the equations $(\ig_1^i)^n = \tg_1(s)^i$,
$(\ig_2^j)^n = \tg_2(s)^j$ instead, we may assume that
$|\ig_1| = |\ig_2|$; let $\alpha = |\ig_1| = |\ig_2|$.
We shall show that the new system of equations admits finitely many solutions, and hence so will the original system.

\begin{claim}
If $\alpha \neq 1$, then there exists an effectively computable constant $n_0 \in \N$, such that $u_n \neq 0$ for all $n\geq n_0$.
\end{claim}
\begin{claimproof}
We may assume that $\alpha > 1$ by inverting the equations if necessary.
Write $P(x,y) = H(x,y) + G(x,y)$ such that $H$ comprises the maximal
(total) degree $d$ monomials of $P$ (and is thus homogeneous), and write
$\ig_1 = \alpha u$, $\ig_2 = \alpha v$, where $|u|,|v| = 1$.
Now $H$ factors into complex lines as it is homogeneous:
$H(x,y) = \prod_{i}(a_i x + b_i y)$, $a_i, b_i \in \overline{\Q}$, so that
$H(x,y) = 0$ if and only if $a_i x + b_i y = 0$ for some $i$.
We now have
\[ |H(\ig_1^n,\ig_2^n)| = (\alpha^d)^n|H((u/v)^n,1)|.\]

We are assuming, in particular, that $\ig_1/\ig_2$ is not a root of unity. We have
$a_i \ig_1^n + b_i \ig_2^n = 0$ for finitely many $n$ and thus $H(\ig_1^n,\ig_2^n)$
vanishes only for finitely many $n$.
Clearly if $|b_i/a_i| \neq 1$ (or either $a_i$ or $b_i$ is zero),
the term $a_i \ig_1^n + b_i \ig_2^n$ does not vanish, and is bounded below in modulus by a constant (for large $n$). Assume then that $b_i/a_i$ has modulus $1$.
Then $|a_i \ig_1^n + b_i \ig_2^n| = |a_i||(\ig_1/\ig_2)^n + b_i/a_i|$.
Applying \cref{cor:bakercor} we have, for all large enough $n$ and for each
$i$, $|a_i||(\ig_1/\ig_2)^n + b_i/a_i| > a/n^c$ where $a$ and $c$ are constants depending on $\tg_1$, $\tg_2$, and $b_i/a_i$.
It follows that for all $n$ large enough $|H(u^n,v^n)| > c_2/n^A$ for some
computable $c_2$, $A$. We deduce that
$|P(\ig_1^n,\ig_2^n)| = D (\alpha^{d})^n/n^A + \mathcal{O}(\alpha^{(d-1)n})$ for
some non-zero constant $D$. Again we have an effectively computable $n_0$
after which no solution can occur.
\end{claimproof}
Again, we may invoke \cref{lemma:fixedn} to search among the finitely
many $n$ which witness a zero in $(u_n)_n$.

Moving along, we consider the case $\alpha = 1$.
\begin{claim}
Assume that $\alpha = 1$ and $\ig_1$ is an algebraic integer.
Then the conclusion of the above lemma holds.
\end{claim}
\begin{claimproof}
Since $\ig_1$ is not a root of unity by assumption, $\ig_1$ has a Galois conjugate $\tilde{\ig}:=\ig_1^{(i)}$ (as in the definition of the height
of $\ig_1$) of modulus larger than $1$. By taking $\sigma$ a Galois conjugation in the field extension of $\Q$ with the elements $\ig_1$, $\ig_2$ and the coefficients of the polynomials of $P$ such that
$\sigma(\ig_1) = \tilde{\ig}$, by relabelling everything under the conjugation, we have an equivalent problem where we assume $|\ig_1|>1$.
(In particular, $\sigma(u_n) = 0$ if and only if $u_n = 0$.)
We may thus conclude as in the previous cases.
\end{claimproof}

To complete the proof of \cref{lemma:twoeigenconstantgammanot}
we assume that that $\ig_1$ has modulus $1$ and is not an algebraic integer. In particular, there exists a prime ideal $\mathfrak{p}$,
effectively computable, such that $\nu_{\mathfrak{p}}(\ig_1) \neq 0$. By replacing
$\ig_1$ by $\ig_2$ if necessary, we may assume
$\nu_{\mathfrak{p}}(\ig_1) > 0$. Let now $\mathfrak{p}$ be any such prime ideal.
Let us write $P(x,y) = x^j R(x,y) + Q(y)$ with
$Q(y) = C \prod_i( y - \beta_i)$ and $j$ is maximal, so that
$R(x,y)$ contains a monomial not involving $x$. Consequently
\begin{align*}
\nu_{\mathfrak{p}}(\ig_1^{jn} R(\ig_1^n,\ig_2^n))
	&= n j \nu_{\mathfrak{p}}(\ig_1) + \nu_{\mathfrak{p}}(R(\ig_1^n,\ig_2^n))\\
	 &\geq n j \nu_{\mathfrak{p}}(\ig_1) - A_1,
\end{align*}
where $A_1$ is a constant,
and
\begin{equation*}
\nu_{\mathfrak{p}}(Q(\ig_2^n))
= \nu_{\mathfrak{p}}(C) + \sum_i \nu_{\mathfrak{p}}(\ig_2^n - \beta_i).
\end{equation*}
In particular, for $n \geq n_0$ with $n_0$ effectively computable, we have that the
second valuation must be proportional to $n$ whenever $P(\ig_1^n,\ig_2^n) = 0$.
For non-zero $\beta_i$, we have by \autoref{cor:p-adiclogform}
$\nu_{\mathfrak{p}}(\ig_2^n-\beta_i) \leq C_i \log n$
for a constant $C_i$ depending on $\lambda_2$, $\beta_i$, and $\mathfrak{p}$.
So if all the $\beta_i$ are non-zero, we have an upper bound on $n$
for which equality can hold.

We conclude that at least one $\beta_i = 0$. Still, to have valuation proportional to $n$, we must have $\nu_{\mathfrak{p}}(\ig_2) \neq 0$ to have arbitrarily large $n$ solving the system. We may repeat this argument for all $\mathfrak{p}$ for which $\nu_{\mathfrak{p}}(\ig_1) \neq 0$. Either we get an effective upper bound on $n$, or $\nu_{\mathfrak{p}}(\ig_1) \neq 0$ if and only if
$\nu_{\mathfrak{p}}(\ig_2) \neq 0$. We deduce that $\ig_1$ and $\ig_2$ sit
over the same prime ideals.
Now if $\nu_{\mathfrak{p}}(\ig_1) = i$ and $\nu_{\mathfrak{p}}(\ig_2) = j$,
consider the equations $\ig_1^n = \tg_1(s)$, $(\ig_2^i/\ig_1^j)^n = \tg_2^i/\tg_1^j(s)$ instead. Now $\nu_{\mathfrak{p}}(\ig_2^i/\ig_1^j) = 0$,
while $\nu_{\mathfrak{p}}(\ig_1) = i$, so that the above argument gives an effective bound on $n$.

This concludes the proof.
\end{proof}

\subsection{Remaining cases of \autoref{subsec:non-constantRank2}}\label{app:mainLemmaCompleted}

We first prove \autoref{lem:dependentModCNonconstant}:

\dependentModCNonconstant*

We need an auxiliary lemma for this.
\begin{lemma}
\label{lem:nBoundOrHsBound}
Consider the equation $\ig(s)^n = \tg(s)$, where neither $\ig$ nor $\tg$ is constant, and let $(n,s)$ be a solution to it. Then either $n \leq n_0$ or
$h(s) < C$ for some constants $n_0$, $C$, depending on $\ig$ and $\tg$.
\end{lemma}
\begin{proof}
Recall from \autoref{lem:1} that we have
\begin{equation*}
 a_1 h(s) - a_2  \leq h(\tg(s)) \leq a_3 h(s) + a_4 \quad\text{ and }\quad
 b_1 h(s) - b_2  \leq h(\ig(s)) \leq b_3 h(s) + b_4
\end{equation*}
for some effectively computable constants $a_i,b_i > 0$.
Let $n_0 = a_3/b_1$, and assume that $h(s) > \max\{a_2/a_1, b_2/b_1\}$ and $n > n_0$. Then
$h(\ig(s)) \geq a_1 h(s) - a_2 > 0$ and thus
\[ n(b_1 h(s) - b_2) \leq h(\ig(s)^n) = h(\tg(s)) \leq a_3 h(s) + a_4. \]
It follows that $h(s) \leq \frac{n b_2 + a_4}{n b_1 - a_3}$ which is bounded above by a constant (as a decreasing function with limit $b_2/b_1$). The claim follows.
\end{proof}

\begin{proof}[Proof of \autoref{lem:dependentModCNonconstant}]
Assume that $\ig_1^{a_1}\cdots \ig_t^{a_t} = c$ identically for some
$c \in \overline{\Q}$. Then $c$ is not a root of unity, as otherwise
$\{\ig_1,\ldots,\ig_t\}$ would not be multiplicatively independent. We obtain the equation\hfill%
\begin{equation}\label{eq:cnstton}
c^n = \tg_1^{a_1} \cdots \tg_t^{a_t}(s).
\end{equation}
If the right-hand-side is also a constant, then there is only one $n$ for
which the equation can hold ($c^n = c^m =d$ implies $c^{n-m}=1$), and this $n$ can be effectively computed as an instance of the one-dimensional Kannan--Lipton Orbit Problem.

If it is not constant, then system \eqref{eq:mainlemmaequations} contains
the equation (after relabelling) $\ig_1(s)^n = \tg_1(s)$ with $\tg_1$ non-constant. Since at least one of the $\ig_j$ is non-constant, we may assume
that both $\ig_1$ and $\tg_1$ are non-constant by considering
$(\ig_1\ig_2(s))^n = \tg_1\tg_2(s)$, where $\ig_2$ is non-constant, if necessary.
For any solution $(n,s)$, we have by \autoref{lem:nBoundOrHsBound} either
$n \leq n_{0}$ or $h(s) < C_i$ for some constant $n_{0} \in \N$, $C_i > 0$. Assuming that
$n > n_{0}$ holds we have the latter bound.
Now there exists a constant $C$ such that $h(\tg_i(s)) < C$ regardless of whether $\tg_i$ is constant or not, applying \autoref{lem:1}.
Consequently, taking heights on both sides of \eqref{eq:cnstton},
we see that $n h(c) = \sum_{i=1}^t |a_i| h(\tg_i(s)) < t\max_{i}\{|a_i|\}C$.
It is evident that $n$ is effectively bounded
above, and the claim follows.
\end{proof}

The remaining cases left from \autoref{subsec:non-constantRank2} to consider are when $\ig_1$, $\ig_2$, $\tg_1$, $\tg_2$ are
multiplicatively dependent, while $\ig_1$ and $\ig_2$ are multiplicatively
independent modulo constants. The proof goes along the proof of \autoref{lem:multiplicativelyindependetTargetsAndEigenvalues}.

\lemmultdepindepmodconst*
\begin{proof}
 Now any multiplicative relation
must involve some $\tg_i$, and without loss of generality
$\tg_2^{a} = \ig_1^{a_1}\ig_2^{a_2}\tg_1^{a_3}$ with $a \neq 0$.
Let us set $c = a_3$ if $a_3 \neq 0$ and $c = 1$ otherwise.
We then have the equations
\begin{align*}
\ig_1(s)^{n c} = \tg_1(s)^{c} \quad\text{ and }\quad\ig_2(s)^{n a} = \tg_2(s)^a = \ig_1(s)^{a_1}\ig_2(s)^{a_2}\tg_1(s)^{a_3}.
\end{align*}
(I.e., if $a_3 = 0$ we keep $\ig_1(s)^n = \tg_1(s)$).

Consider the family of pairs of multiplicative relations $\vec{a}_n := (n c, 0 , -c)$ and
$\vec{b}_n = (-a_1, n a - a_2,-a3)$. Clearly, if the $\vec{a}_n$ and $\vec{b}_n$ are collinear,
then $n = a/a_2$. So, save for this exceptional $n$,
the the multiplicative relations $\vec{a}_n$ and
$\vec{b_n}$ are independent for any $n\geq 1$. (For the claim, we note we can take $n_0 \geq a_2/a$.)

Assume first that $\ig_1$, $\ig_2$, $\tg_1$ are multiplicatively
independent. Consider the curve $\mathcal{C}$ defined by these functions (similar to
the construction in the proof of \autoref{lem:multiplicativelyindependetTargetsAndEigenvalues}), and let $\mathcal{C}'$ be an absolutely
irreducible component of it. If the functions are multiplicatively
independent modulo constants, we conclude, as in the first part of the proof of \autoref{lem:multiplicativelyindependetTargetsAndEigenvalues}, utilising \autoref{thm:BMZfinManyPoints} for all $n \neq a_2/a$.

If the functions are multiplicatively dependent modulo constants, we may
apply \autoref{thm:BMZ2finManyPoints} as in the second part of the the proof of \autoref{lem:multiplicativelyindependetTargetsAndEigenvalues}.

We are left with the case that $\ig_1$, $\ig_2$ and $\tg_1$
are multiplicatively dependent and we have
$\tg_1^b = \ig_1^{b_1} \ig_2^{b_2}$ with $b \neq 0$.
We again get the equations
\begin{align*}
\ig_1(s)^{n c b} = \tg_1(s)^{c b} = \ig_1(s)^{b_1 c} \ig_2(s)^{b_2 c}
\quad\text{ and }\quad
\ig_2(s)^{n a b} &= \ig_1(s)^{b a_1}\ig_2(s)^{b a_2}\tg_1(s)^{b a_3} \\
&= \ig_1(s)^{b a_1 + a_3 b_1}\ig_2(s)^{a_2 b + b_2 a_3}.
\end{align*}

Recall now that $\lambda_1$ and $\lambda_2$ are multiplicatively independent. Putting all on one side, we get the equations
\begin{align*}
1 = \ig_1(s)^{n c b - b_1c} \lambda_2(s)^{-b_2 c}\quad\text{ and } \quad
1 =  \ig_1(s)^{-b a_1 - a_3 b_1} \ig_2(s)^{n a b - a_2 b}.
\end{align*}
Notice that now neither $c b$ nor $a b$ equals $0$ according to our choices. Let now $\vec{a}_n = (n c b - b_1 c,-b_2 c)$ and $\vec{b}_n = (-b a_1 - a_3 b_1, n a b - a_2 b)$. The matrix with rows $\vec{a}_n$ and $\vec{b_n}$ has determinant quadratic in $n$. Hence there are at most two exceptional values of $n$ for which the vectors are collinear. Otherwise $\vec{a}_n$ and $\vec{b}_n$ are $\Z$-linearly independent. It is evident that, save for the at most two exceptional
values of $n$, the multiplicative relations $\vec{a}_n$ and $\vec{b}_n$ are $\Z$-linearly independent. Hence for $n$ not an exceptional value, there are finitely many points $s$
for which the equations can be satisfied. (Again, for the claim, we may take $n_0$ larger than both of the two exceptional values of $n$.) On the other hand, we may solve the problem for the exceptional values of $n$ using \autoref{lemma:fixedn}.

Consider again the curve defined by $\ig_1$ and $\ig_2$ similar to the above, and any of its absolutely irreducible components.
As $\ig_1$ and $\ig_2$ are multiplicatively independent modulo constants,
we may apply \autoref{thm:BMZfinManyPoints} to conclude as above.
This concludes the proof.
\end{proof} %

\section{Additional Material for Section \ref{sec:rank1case}}
\label{appen:rank1}
\subsection{W.l.o.g. there is a single equation}
\label{sec:rank1relations}

\wlogsingle*
\begin{proof}

Recall that when $\rank\{\ig_1,\ldots,\ig_t\} = 1$, we may replace the system of equations with a system consisting of one equation
$\ig(s)^n = \tg(s)$. The process might involve creating new solutions that
do not solve the original system. We take care of this problem by showing how to recover solutions to the main system from solutions to the single
equation system.

Assume first that there exists a non constant eigenvalue
$\ig$ attaining non-real values. Then, by assumption, also its complex
conjugate $\overline{\ig}$ is an eigenvalue. As
$\rank\{\ig,\overline{\ig}\} = 1$, we have $\ig^a = \overline{\ig}^b$ with
$a$, $b$ non-zero since neither is assumed to be a root of unity.
If $a=1,b=-1$ (or visa-versa) then $\lambda =(\overline{\lambda})^{-1}$, hence $|\lambda| = \frac{1}{|\lambda|}$, thus $|\lambda| = 1$. 

If $a = b$ then $\lambda^{a} = \overline{\lambda}^{a} = \overline{\lambda^{a}}$, thus $\lambda^a$ is real (and so $\lambda^{2a}$ is positive and real). This case was eliminated already by \cref{lemma:nonzeroonly1}.

In the remaining $a\ne b$ and $|a| > 1$: we get $\ig^{b+a} = |\ig|^{2b}$, and taking absolute values both
sides, we get $|\ig|^{b+a} = |\ig|^{2b}$, which occurs only when $|\ig| = 1$ identically.
Therefore the values of $\ig$ lie on the unit circle.

Assume that the system contains also a real-valued function $\ig_1$.
We similarly have $\ig_1^{c}  = \ig^{d}$ with $c$ and $d$ not zero.
Taking again absolute values on both sides, we have that $|\ig_1|^c = 1$.
It follows that $\ig_1 = \pm1$, and we therefore have $\ig_1$ is a constant root of unity. But, we may remove such an eigenvalue from the analysis
by \autoref{lemma:cleanuplemma}. We may thus assume that either all
eigenvalues are real-valued, or are complex-valued with values on the unit circle.

Assume first that all the eigenvalues of the system are real-valued (and not constant $\pm1$). We show that we may assume there exists a function
$\mu$, not necessarily any one of the eigenvalues, such that
$\ig_i = \mu^{b_i}$ for each $i$. If there is only one such eigenvalue,
there is nothing to prove, so assume that there are several. Partition the
domain in intervals such that in each interval, the $\ig_i$ and $\tg_i$ have constant sign. We first show that we
may assume they are both positive. Indeed, if in any interval we have
$\ig_i$ positive and $\tg_i$ negative, there can be no solutions.
If $\ig_i$ is negative and $\tg_i$ is positive, then there can only
be solutions with $n$ even. Therefore, we may replace the equations
by
$\ig_i^2(s)^{n_1} = \tg_i(s)$, with $n_1 \in \N$ without creating spurious solutions. Here both $\ig_1^2(s)$ and $\tg_i$ are positive.
Similarly, if both $\ig_i$ and $\tg_i$
are negative, there can only be a solution for odd $n$. We may therefore
replace the equations with
$\ig_i^{2}(s)^{n_1} = \tg_i/\ig_i(s)$, where $n_1 \in \N$. No new solutions are created in this process, while $\ig_1^2$ and $\tg/\ig$ are both positive.

We may from now on consider one of the finitely many intervals in the above partition. For each pair $\ig_1,\ig_2$, we have $\ig_1^{a_i} = \ig_i^{b_i}$ for some non-zero $a_i$ and $b_i$. Recall that we also have $\tg_1^{a_i} = \tg_i^{b_i}$ by assumption (otherwise $\tg_1^{a_i}(s) = \tg_2^{b_i}(s)$
holds for at most finitely many $s$, deeming the problem decidable).
Take $\mu = \ig_1^{1/\ell}$ and $\eta = \tg_1^{1/\ell}$ where
$\ell = \lcm_i(b_i)$. This is well-defined as the $\ig_i$ and $\tg_i$ are positive. Then for each $i$ we have
$\ig_i = \ig_1^{a_i/b_i} = \mu^{\ell_i}$ and similarly
$\tg_i = \eta^{\ell_i}$, for some integer $\ell_i$.
Now any solution of $\mu^n(s) = \eta(s)$ is a solution to the whole
system, and it thus suffices to search for solutions for this single equation.

We then turn our attention to the case of eigenvalues attaining non-real
values. As pointed out above, the values of the eigenvalues lie on the
unit circle. Assume that $\ig_1$ is such. Recall that for each $\ig_i$ we
have non-zero $a_i$, $b_i \in \Z$ such that $\ig_1^{a_i} = \ig_i^{b_i}$ and $\tg_1^{a_i} = \tg_i^{b_i}$.
Partition the domain into many finitely intervals according to the points
where the non-constant $\ig_1^{a_i}$, $\ig_i^{b_i}$, $\tg_1^{a_i}$, and
$\tg_i^{b_i}$ attain the value $-1$. (If some $\tg_i$ is constant $-1$ we do not take this into consideration when defining the intervals. Also, by assumption none of the $\ig_i$ are constant $-1$ as this is a root of unity). Let $\Log$ be the principal branch of the complex logarithm function, and for $a \in \N$, $a \geq 1$, define $z^{1/a}:= \exp(1/a \Log z)$. Notice that the function is not continuous for $z \in \C$, but in each of the intervals constructed above, the functions $\ig_i^{1/a}$ are continuous and single-valued. We focus on one of the intervals from now on.
We show that there exist algebraic functions $\mu,\eta$, integers $\ell_i$,
and $b_i$th roots of unity $\omega_i$, $\omega_i'$
such that $\ig_i = \omega_i \mu^{\ell_i}$ and $\tg_i = \omega_i' \eta^{\ell_i}$ for each $i$.
Let $\ell = \lcm_i(b_i)$ and set $\mu = \ig_1^{1/\ell}$ and $\eta = \tg_1^{1/\ell}$. Then $\ig_1 = \mu^{\ell}$, $\tg_1 = \eta^{\ell}$,
and $\mu^{\ell a_i} = \ig_1^{a_i} = \ig_i^{b_i}$. Similarly $\eta^{\ell a_i} = \tg_i^{\ell_i}$. It follows that
$\ig_1 = \omega_i \mu^{\ell_i}$ for some $\omega_i$ a $b_i$th root of unity, and $\ell_i = a_i \ell/b_i \in \Z$. Indeed, for any $s$ we have
$\ig_i(s) = \omega_s \mu^{\ell_i}(s)$ for some $b_i$th root of unity
$\omega_s$. By continuity, $\omega_s$ is also continuous, and hence is constant. Similarly $\tg_i = \omega_i' \eta^{\ell_i}$, as desired.

The equations are now equivalent to
\[
(\omega_i \mu^{\ell_i}(s))^n = \omega_i'\eta^{\ell_i}(s)\qquad i = 1,\ldots,t.
\]
Considering the subsequences $n = r \ell + m$, $r \in \N$, for $m = 0,\ldots, \ell-1$,
we may consider the equations
\[
\mu^{\ell_i}(s)^n = \omega_i''\eta^{\ell_i}(s), \qquad i = 1,\ldots,t,
\]
where $\omega_i'/\omega_i^m$ has been combined into $\omega_i''$, yet another $b_i$th root of unity, with $\omega_1'' = 1$.

The solutions to $\ig_1(s)^n = \tg_1(s)$ are in one-to-one correspondence
to the union of the solutions to $\mu(s)^n = \omega\eta(s)$ where $\omega$
ranges over the $\ell$th roots of unity. Assuming $(n,s)$ is a solution
to $\mu(s)^n = \omega\eta(s)$, we get
$\mu^{\ell_i}(s)^n = \omega^{\ell_i}\eta^{\ell_i}(s)$ for each $i$. We thus deduce that the system of equations has
a solution if and only if $\mu(s)^n = \omega\eta(s)$ for some $\ell$th root of unity $\omega$ such that $\omega^{\ell_i} = \omega_i''$ for each $i=2,\ldots,t$. It is plain to check whether the $\omega_i''$ satisfy such a relation, so it suffices to characterise the solutions to
$\mu(s)^n = \omega\eta(s)$, $\omega$ any one of the $\ell$th roots of unity.
\end{proof}

\subsection{Real case}

\begin{figure*}[ht]
\centering
\includegraphics[width=0.7\linewidth]{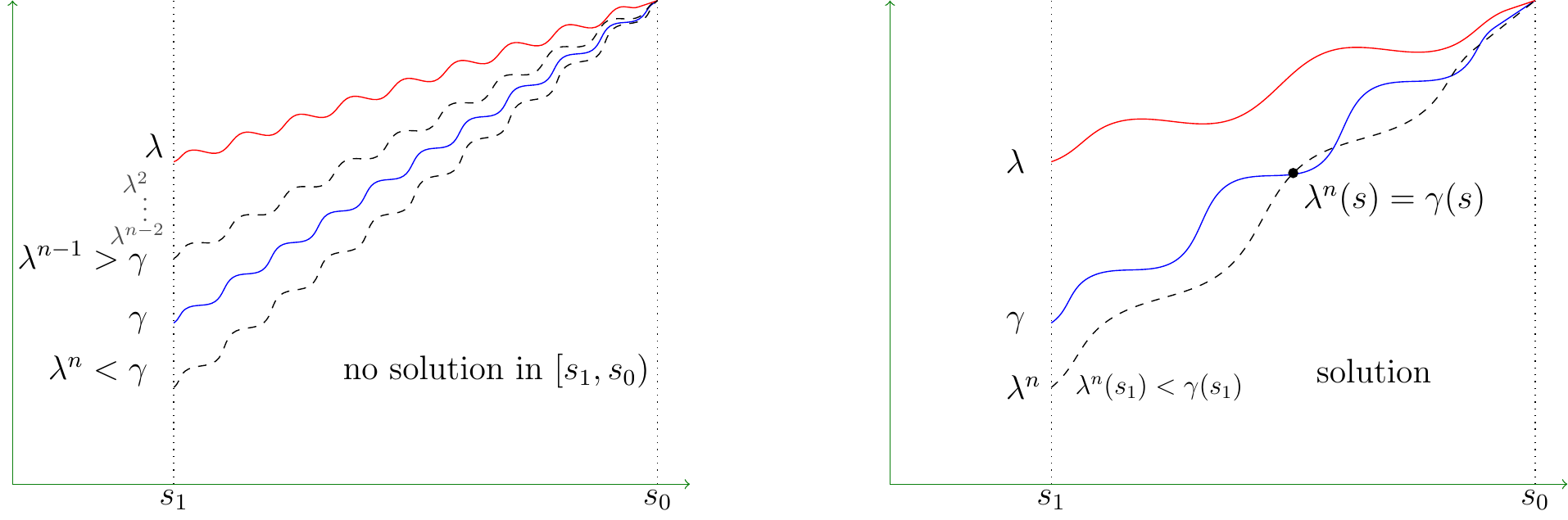}
\caption{Cases for $\lambda(s) \to 1$ as $s\to s_0$.}
\label{fig:casesrank1}
\end{figure*}

\rankonereallemma*
\begin{proof}
For real-valued functions $f$, $g$, if $f(\param) < g(\param)$ for all $\param$ in a set $E$, we use the notation $f < g$ (over $E$).

Let us consider the partition of $\mathbb{R}\setminus \exceptionalpoints{}$ into interval subsets $S_1,\dots,S_i$, such that for each subset either  $0 \le |\lambda| <1$,  or $|\lambda| > 1$, with the finite set of points  $\{\point : \lambda(\point) = 1\}$ excluded and handled separately (recall, by \autoref{lemma:nonzeroonly1} and \autoref{lemma:cleanuplemma} $\lambda$ is not constant 0 or 1). We will focus on the subsets where $|\lambda| \le 1$. Given such a subset $S_i$, we only need to consider each interval $D \subseteq S_i$ where $\{s\mid 0 \le|\gamma(s)| \le 1\}$. The remaining case where $|\lambda| > 1$ reduces to our case by considering $\frac{1}{\lambda^n}(\param) = \frac{1}{\gamma}(\param)$. Note that this partition is finite as the function $|\lambda(\param)| = 1$  at only finitely many points (similarly for $|\gamma(\param)| = 1$), and these points can be checked explicitly.

First let us consider $\lambda$ constant, and we may assume  $0 < \lambda < 1$. Then compute $a = \inf_\param \tg_i(\param)$ and $b = \sup_\param \tg_i(\param)$ and decide whether there exists $n$ such that $a \le \ig_i^n \le b$. Henceforth, $\lambda$ is not constant.

Whilst we assume $|\lambda(\param)| < 1$, still $|\lambda(\param)|$ could be arbitrarily close to $1$. We first consider the subset of $D$ where this is not the case. Let $\delta > 0$ be a small rational number, and consider
the set $\mathcal{S}_{\lambda}(\delta) \subseteq D$ comprising those $s$ such
that $|\lambda(\param)| < 1-\delta$. Then, for each $s\in \mathcal{S}_{\lambda}(\delta)$ we
have $|\lambda^n(\param)| < (1-\delta)^n$ for all $n \geq 0$. In particular, $\lambda^n(s)$
tends to $0$ exponentially.

Similary bounding $|\gamma|$ away from $0$, let $\mathcal{S}_{\gamma}'(\delta')$, for $\delta'>0$ a small rational number,
comprise those points $s \in \mathcal{S}_{\lambda}(\delta)$ for which  $|\gamma(s)| > \delta'$.

Then, for $n$ larger than $\log(\delta')/\log(1-\delta)$, we have $|\lambda^n| < |\gamma|$, leading to the lemma:

\begin{lemma}
Let $\delta,\delta' > 0$ be fixed small rational numbers. Then there exists
$n_{\delta,\delta'} \in \mathbb{N}$ such $\lambda^n(s) = \gamma(s)$ does
not have a solution with $n \geq n_{\delta,\delta'}$ and
$s \in \mathcal{S}_{\lambda,\gamma}(\delta,\delta') = \mathcal{S}_{\gamma}'(\delta') \cap \mathcal{S}_{\lambda}(\delta) $.
\end{lemma}

Hence, given $\delta,\delta'$ and having computed $n_{\delta,\delta'}$, solutions for each $n\le n_{\delta,\delta'}$ can be found by \autoref{lemma:fixedn}.

Recall, without loss of generality we assume $\lambda,\gamma$ are positive, if necessary by taking even or odd sub-sequences.  Hence the remaining cases for $s \in D \setminus \mathcal{S}_{\gamma}'(\delta')$, that is when $\lambda(s)$ is approaching 1, or $\gamma(s)$ is approaching 0.

We will make repeated use of the following %
immediate consequence of the intermediate value theorem
\begin{lemma}\label{lemma:ivt}
Given two continuous functions $f,g$ on the interval $[a,b]$ with $f(a) < g(a)$ and $f(b) >g(b)$, there exists $\point$ such that $f(\point) = g(\point)$.
\end{lemma}
and its  immediate corollary:
\begin{corollary}\label{corollary:ivt}
Given two continuous functions $f,g$ on the interval $(a,b)$. One of the following occurs
\begin{itemize}
	\item $f(x) > g(x)$ for all $x\in(a,b)$, or
	\item $f(x) < g(x)$ for all $x\in(a,b)$, or
	\item there exists $\point\in(a,b)$ such that $f(\point) = g(\point)$
\end{itemize}
\end{corollary}
\begin{proof}
Suppose there exists $x,y\in(a,b)$ such that $f(x) > g(x)$ and $f(y) < g(y)$, then on the interval $[x,y]\subseteq (a,b)$ there exists $\point$ such that $f(\point) = g(\point)$.
\end{proof}

We assume that $\delta,\delta'$ are chosen giving  interval $D \setminus \mathcal{S}_{\lambda,\gamma}(\delta,\delta')$. Let $E$ be one such interval with problematic endpoint $s_0$, that is $E = (s_0,s_1]$ or $E = [s_1,s_0)$.  We assume we choose $\delta,\delta'$ small enough so that $\lambda(\param)$ and $\gamma(\param)$ are monotonic in $E$. This is because the derivative of an algebraic function is an algebraic function\footnote{Differentiating the polynomial defining $\ig$ implicitly with respect to $\param$, we get a polynomial $P(s,\ig(\param),\ig'(\param))$. Eliminating with respect to $\ig(\param)$, we get a polynomial relation with $s$ and
$\ig'(\param)$}, and therefore has finitely many roots, thus the function changes direction finitely many times.  Furthermore, it is evident that such $\delta$, $\delta'$ are effectively computable.

Let us start with the case that $\lambda(\param) \to 1$ as $\param \to s_0$.

First, let us assume there exists $b_1,b_2$ such that $ 0 < b_1 < \gamma < b_2 < 1$ over $E$, then since $\lambda(s_1) <1$ we have $\lambda(s_1)^n < b_1$  for some $n$ (and $\lambda(\param)^n \to 1 > b_2$ as $\param \to s_0$). Hence by \autoref{lemma:ivt}, there is a solution $\lambda^n(s) = \gamma(s)$ at some point $s \in E$. Clearly $n$ is computable, and we may compute a suitable $s$ for which equality holds.

Otherwise we have $\gamma(\param)$ is also approaching $1$ or $0$ as $\param \to s_0$. Let us start with $1$: It must be the case, by \autoref{corollary:ivt}, that either $\gamma < \lambda$ or $\lambda < \gamma$ in $E$, otherwise there is a point $\point$ such that $\lambda(\point) = \gamma(\point)$ and the answer is \textsc{yes} (in fact, at $n = 1$). If $\lambda < \gamma$ then the answer is \textsc{no}, as $\lambda^n < \lambda < \gamma$ over $E$. Hence we must consider $\lambda > \gamma$ and so $1 > \lambda(s_1) > \gamma(s_1)$.

Then we can compute $n$ such that $\lambda(s_1)^n < \gamma(s_1)$. After this occurs either there exists $\point$ such that $\lambda(\point)^n= \gamma(\point)$, or $\lambda^n < \gamma$ and so we only need to check every $m \le n$ (via, \autoref{lemma:fixedn}). These two cases are depicted in \autoref{fig:casesrank1}.

Now let us assume $\gamma(\param) \to 0$ as $\param\to s_0$. Similarly we assume monotonicity of $\lambda,\gamma$ as $\param \to s_0$. Again we have $\lambda > \gamma$ over $E$ (otherwise $\lambda(\point) = \gamma(\point)$ at some $\point$, answer \textsc{yes}, or $\lambda^n <\lambda < \gamma$, answer \textsc{no}). Again we search for $n$ such that $\lambda(s_1)^n < \gamma(s_1)$, at which point either there exists $\point$ such that $\lambda(\point)^n = \gamma(\point)$ or  $\lambda^n < \gamma$ over $E$ and hence $\lambda^m < \gamma$ for all $m \ge n$ (it remains to check each $1,\dots, n$ manually, via \autoref{lemma:fixedn}).
\end{proof}

\subsection{Non-real case}
We prove
\rankOneComplex*

\begin{proof}[Proof of \cref{lem:rankOneComplex}]
Since $\ig$ is of constant modulus $1$, we are only concerned with points where $\tg$ is of modulus $1$. If $\tg$ is not of constant modulus $1$, then there are only finitely
many $s$ for which $\tg$ intersects the unit circle and only these points need to be checked

Assume first that $\ig$ is a constant.
If $\tg$ is of constant modulus $1$, then the range of $\tg$ defines (possible several) open arcs on the unit circle. The orbit of $\ig$ is dense on the unit circle, as it is not assumed to be a root of unity. Therefore, there exist (infinitely many) integers $n$ such that
$\ig^n$ hits such an arc. Such an $n$ can be straightforwardly computed,
after which the suitable $\point$ can be computed. The single equation
therefore always has a solution.

Otherwise, we may assume that $\ig$ and $\tg$ define continuous arcs on the circle. Furthermore, we may assume that the arcs do not cross the line $(-\infty,0]$. (In case $\tg$ is constant, it defines a point.) Let us write
$\ig$ and $\tg$ in polar form: $\ig = \exp(\iu \theta)$, $\tg = \exp(\iu \psi))$, where now $\theta,\psi \colon D \to [-\pi,\pi)$ are continuous, and $\iu$ is the imaginary unit.
The derivative of an algebraic function is algebraic, here it is
$\iu\theta'(\param) \exp(\iu\theta(\param))$. We deduce that $\theta'(\param)$ is an algebraic function, and the zeros of it may be computed. We may define an interval in which $\theta$ and $\psi$ are monotone: they draw continuous arcs on the unit circle and are rotating in one direction with $s$ varying. Compute some approximations $\theta_0$, $\psi_0$ of the length of the arcs, and compute $n$ so large, so that $n\theta_0 > 4\pi + \psi_0$ (notice that the
$n\theta_0$ gives an approximation for the length of the arc defined by
$\ig^n$). So, while $s$ ranges over the interval, the arc of $\ig^n$ winds
around the unit circle at least twice. By the intermediate value theorem
there must be a point at which $\ig^n(\point) = \tg(\point)$. To see this, map the progress of the arc onto the real line. Let the endpoints of the interval be $s_0$ and $s_1$. Assume $\theta(s_0) < \psi(s_0) \leq \theta(s_0) + 2\pi$ (if not, add integer multiples of $2\pi$ to $\psi(s_0)$).
 Now
$n\theta(s_1) \geq \theta(s_0) + 4\pi \geq \psi(s_0) + 2\pi > \psi(s_1)$. Consequently, by the intermediate value theorem, there must be
a point where the values $n\theta$ and $\psi$ coincide, as they are continuous functions.
\end{proof}

\end{document}